\documentclass[11pt,twoside]{article}
\usepackage[utf8]{inputenc}
\usepackage[english]{babel}
\usepackage{fancyhdr}
\usepackage{amssymb,amsmath,amsthm}
\usepackage{mathtools}
\usepackage{algorithm}
\usepackage{algorithmic}
\usepackage{graphicx}
\usepackage{psfrag}
\usepackage{hyperref}
\usepackage{multirow}
\usepackage{bigstrut}
\usepackage{rotating}
\usepackage{caption}
\usepackage{subfig}
\usepackage{bbold}
\usepackage[inline]{enumitem}
\usepackage[square,numbers]{natbib}

\newcommand{\R}{\mathbb{R}}
\newcommand{\vol}{\operatorname{vol}}
\newcommand{\one}{\mathbb{1}}

\newtheorem{definition}{Definition}
\newtheorem{lemma}{Lemma}
\newtheorem{proposition}{Proposition}
\newtheorem{theorem}{Theorem}
\newtheorem{corollary}{Corollary}

\def\sign{\mathrm{sign}}
\DeclarePairedDelimiter{\norm}{\lVert}{\rVert} 
\DeclarePairedDelimiter{\abs}{\lvert}{\rvert} 
\DeclareMathOperator*{\argmax}{Argmax} 

\def\FASTATVO{\texttt{FAST-ATVO}}
\def\SWAP{\texttt{SWAP}}

\hypersetup{
    colorlinks=false,
    urlcolor=black,
    pdfborder={0 0 0}
}

\newcommand{\email}[1]{E-mail: \href{mailto:#1}{\texttt{#1}}}

\begin{document}
\thispagestyle{plain}

\setcounter{page}{1}

{\centering
\textbf{\LARGE Total variation based community detection using a nonlinear optimization approach}

\bigskip\bigskip
Andrea Cristofari$^*$, Francesco Rinaldi$^*$, Francesco Tudisco$^\dag$
\bigskip

}

\begin{center}
\small{\noindent$^*$Department of Mathematics ``Tullio Levi-Civita'' \\
University of Padua \\
Via Trieste 63, 35121 Padua (Italy) \\
E-mail: \texttt{andrea.cristofari@unipd.it}, \texttt{rinaldi@math.unipd.it} \\
\bigskip
$^\dag$Gran Sasso Science Institute \\
Viale F. Crispi 7, 67100 L'Aquila (Italy) \\
\email{francesco.tudisco@gssi.it} \\
}
\end{center}

\bigskip\par\bigskip\par
\noindent \textbf{Abstract.}
Maximizing the modularity of a network is a successful tool to identify an important community of nodes. However, this combinatorial optimization problem is known to be NP-complete. Inspired by  recent nonlinear modularity eigenvector approaches, we introduce the modularity total variation $TV_Q$ and show that its box-constrained global maximum coincides with the maximum of the original discrete modularity function.
Thus we describe a new nonlinear optimization approach to solve the equivalent problem leading to a community  detection strategy based on $TV_Q$. The proposed approach relies on the use of a fast first-order method
that embeds a tailored active-set strategy. We report extensive numerical comparisons with standard matrix-based approaches and the Generalized RatioDCA approach for nonlinear modularity eigenvectors, showing that our new method compares favourably with state-of-the-art alternatives. 


\bigskip\par
\noindent \textbf{MSC2000 subject classifications.} 49M20. 65K10. 91D30. 91C20.

\section{Introduction}
Identifying important communities in a complex network is a difficult and  highly relevant problem which has
applications in different disciplines, including  social  network  analysis  \cite{mercado2016signed,ozer2016community},  molecular  biology  \cite{hartwell1999molecular},
politics \cite{porter2009communities}, computer science \cite{schaeffer2007graph} and many more (see e.g. \cite{newman2010networks}).  Community detection can be used as a way to highlight so--called  mesoscale  properties of a network and further investigation into such communities may allow to gain insights about the
network structure and the behavior of processes that take place on the network.

There is a large literature on the subject, with numerous different definitions of community and associated community detection methods \cite{fortunato2010community,lancichinetti2009community}.
In this work a community is roughly understood as a set of nodes being highly connected inside and poorly connected with the rest of the graph.
Although this high-level notion of  community is fairly easy to understand and broadly shared in the literature, a  great variety of rigorous definitions have been proposed to make the community detection task mathematically and algorithmically precise, including notions involving edge-counting, random walk trapping, information theory and generative models such as stochastic block models, see e.g.\ \cite{fortunato2016community}.  A  popular and successful idea, originally introduced by Newman and Girvan in \cite{newman2004finding},  is based on the optimization of the modularity quality function.

The modularity function of a set of  nodes~$S$ quantifies the difference between the actual and expected weight of edges in $S$, if edges were placed at random according to a random \textit{null model}. The set $S$ is then typically identified as a community  if its modularity is ``large enough''. Thus, the definition of modularity function depends on the choice of the random model to which the actual network is compared to. Although various choices for such a null model have been considered in the literature \cite{fasino2016generalized,traag2011narrow}, the arguably most popular and successful choice is based on the configuration model or Chung-Lu random graph \cite{chung2006complex}. This is the null model we will always assume throughout this work.

Before providing the definition of modularity, let us introduce the notation and terminology that will be used throughout the whole paper.

Given a vector $x$, we denote by $x_i$ the $i$th component of $x$.
Given a matrix $M$, we denote by $M_{ij}$ the entry of $M$ having position $(i,j)$.
Similarly, given a vector $x$ and a set of indices $I \subseteq \{1,\ldots,n\}$, we indicate by $x_I$ the subvector of $x$ with components $x_i$, $i \in I$.
Further, the gradient of a function $f \colon \R^n \to \R$ at $x$ is a column vector denoted by $\nabla f(x)$.

Now, let $G=(V,E)$ be an undirected network with node set $V=\{1,\dots,n\}$ and nonnegative weight matrix $A = (A_{ij}\geq 0)$ such that $A_{ij}>0$ if and only if $ij\in E$.
Let $d$ be the vector of (weighted) degrees, $d_i = \sum_{j}A_{ij}$, and let $\vol G = \sum_{i=1}^n d_i$ be the graph volume.
\begin{definition}[Modularity of a set]\label{def:Q(S)}
Let $S\subseteq V$ be any set of nodes. The  quantity
$$
Q(S) = \frac 1 {\vol G}\sum_{i,j\in S}\Big( A_{ij}- \frac{d_id_j}{\vol G}\Big)
$$
is the modularity measure of the set $S$.
\end{definition}
The rank-one term $d_id_j/\vol G$ is the one responsible for the null model and it has to be interpreted as the expected weight of edges in the graph, if edges were placed  independently at random according to the Chung-Lu model. The term $A_{ij}$, instead, is the actual weight of the edges and so, when $Q(S)>0$, the weight of edges in $S$ exceeds the expected weight. We call such an $S$ a module. Modules with large value of modularity are typically good indicators of communities in a network.
On the other hand, not all the subsets with positive modularity form a community and the range of attainable values of the modularity function may vary significantly from network to network, see e.g.~\cite{cinelli2019network,guimera2004modularity}. While discerning whether a module is or is not a community is itself an interesting research question, this issue goes beyond the scope of this work. Here we  assume that the ``most important'' module in $G$ is a decent indicator of a community and we formally state the following leading module problem:
\begin{definition}[Leading module problem]\label{def:leading_module}
Find $S^*\subseteq V$ such that
 \begin{equation}\label{eq:leading_module}
  Q(S^*) = \max_{S\subseteq V}Q(S) =: Q^* \,.
\end{equation}
\end{definition}
Note that, with Definition \ref{def:Q(S)}, if $S$ is a module, then also its complement $\bar S = V\setminus S$ is so. In fact, a simple computation reveals that
$$
Q(\bar S) = Q(S)
$$
for all $S\subseteq V$. In particular, both $Q(V)$ and $Q(\emptyset)$ are zero. So the leading module problem can be equivalently thought of as a leading bipartition problem, where one assesses the bipartition as a whole rather than the individual community. Already this simple observation shows that the solution to \eqref{eq:leading_module} is not unique. In general it is well known that very dissimilar sets $S$ may yield optimal (or nearly optimal) modularity values $Q(S)$ \cite{good2010performance}, and we are interested in computing any of them.

From Definition \ref{def:leading_module}, the  leading module  problem boils down to a combinatorial optimization problem which is known to be strongly NP-complete \cite{brandes2007finding}.
Different strategies have been proposed to compute an approximate solution. Linear relaxation approaches are based on the spectrum of specific matrices (as the modularity matrix or the Laplacian matrix) and have been widely explored and applied to various research areas, see e.g. \cite{fasino2014algebraic,pmlr-v84-mercado18a,newman2006finding,shen2010spectral}. Computational heuristics have been developed for optimizing directly the discrete quality function (see e.g. \cite{lancichinetti2009community,porter2009communities}), including for example greedy algorithms \cite{clauset2004finding}, simulated annealing \cite{guimera2004modularity} and  extremal optimization \cite{duch2005community}. Among them, the locally greedy algorithm known as Louvain method \cite{blondel2008fast} is arguably the most popular one.

However, most of the times, the proposed approaches are meant to address the multi-community problem where, rather than just $Q$, one aims at maximizing the sum $\sum_{k}Q(S_k)$ over the set of all possible partitions $\{S_1,S_2,\dots,\}$ of $V$. As the number of sets forming the partition is not initially prescribed, this problem is intrinsically different from the leading module problem and adapting the various approaches is most of the times either not trivial or not possible.
A model that is gaining increasing popularity in recent years and that adapts very well to the leading module case is based on nonlinear relaxation. This approach exploits the connection between modularity optimization and nonlinear eigenvalue problems \cite{buhler2009spectral,hein2010inverse,tudisco2016nodal,tudisco2018core}, on the one hand, and with total variation and image processing \cite{boyd2017simplified,bresson2013multiclass,hu2013method}, on the other. In particular, it is proven in \cite{tudisco2018community} that the optimization problem \eqref{eq:leading_module} can be equivalently recast into the optimization of a function of the type
\begin{equation}\label{eq:ratio}
r(x) = \frac{f(x)}{g(x)}
\end{equation}
where $f$ and $g$ are suitable continuous functions of real variables. As the two problems of maximizing $Q$ and maximizing $r$ coincide, this process is called \textit{nonlinear exact relaxation} of $Q$ \cite{tudisco2018community}. The advantage of this approach is that, when a local maximum of $r$ is computed, the corresponding maximizer can be used to  identify a community of nodes in the network that consistently  outperforms standard approaches based on non-exact relaxations. 

Furthermore, this result shows that one can recast the modularity combinatorial optimization problem \eqref{eq:leading_module} as a nonlinear eigenvector problem and approach it using strategies for the solution of nonlinear eigenproblems. 
In particular, inspired by the inverse power method for matrices, an algorithm called Generalized RatioDCA is proposed in \cite{tudisco2018community} for optimizing the ratio of functions in \eqref{eq:ratio}. Such method is an extension of the RatioDCA method originally developed in the context of unsupervised learning \cite{hein2011beyond}. Each iteration of  Generalized RatioDCA consists of two substeps that roughly go as follows: the first step (outer iteration) is based on the observation that any critical point $x$ of $r$ is such that $0\in \partial f(x) -\lambda \partial g(x)$, for some coefficient $\lambda\in \mathbb R$ and where $\partial$ denotes the subgradient. This, combined with the properties of the subgradient, allows us to 
transform the optimization of $r(x)$ into the optimization of a new function $\tilde r(x)$ being defined as the difference of convex functions; in the second step (inner iteration) the new function $\tilde r$ is optimized via a primal-dual algorithm such as FISTA \cite{beck2009fast} or PDHG \cite{chambolle2011first}. A  convergence result is proposed in \cite{tudisco2018community}, showing that the generated sequence eventually converges to a stationary point of $r$.
The main drawback of this approach is that the inner-outer optimization method proposed can be computationally demanding and thus makes this approach somewhat prohibitive for large datasets.

The contribution of this work is twofold:

\textbf{1.} We show, in Section \ref{sec:TV}, that  the leading community problem \eqref{eq:leading_module} is equivalent to the box-constrained optimization of the so--called ``modularity total variation'', a particular graph total variation function with positive and negative weights. This result, presented in Theorem \ref{thm:main}, improves the one of \cite{tudisco2018community} and shows that \eqref{eq:leading_module} is equivalent to the problem of maximizing a continuous real-valued function $f(x)$ subject to an arbitrary box constraint of the type $-a \le x_i \le b$, $a,b>0$. Moreover, it provides a simpler and self-contained proof of such equivalence result.
In other words, this result unveils the intriguing equivalence between the unconstrained combinatorial leading module optimization problem and the constrained but continuous optimization of $f(x)$ and thus it allows us to efficiently attack the leading module problem with techniques from continuous optimization.

\textbf{2.} Based on this theorem, in Section \ref{sec:algorithm} we propose an approach to optimize the modularity total variation function $f$. Precisely,
we first modify the nonsmooth total variation by considering an approximation $\tilde f$ with continuous derivatives, and then we maximize it using an algorithmic framework that embeds an active-set first-order method for box-constrained problems, named
Fast Active-SeT based Approximate Total Variation Optimization algorithm (\FASTATVO). We prove global convergence of \FASTATVO\
and provide extensive experiments which show that the computational performance and modularity quality score obtained with
the proposed approach compare favorably with the classical linear spectral method \cite{newman2006finding} and the Generalized RatioDCA method described in \cite{tudisco2018community}. The  \FASTATVO\ software that we developed is publicly available, as pointed out in Section~\ref{sec:experiments}.

The paper is structured as follows: In the next section we introduce the concept of modularity total variation, we review the linear and nonlinear spectral methods for modularity maximization and we prove our first main result, showing the equivalence between \eqref{eq:leading_module} and a continuous box-constrained optimization problem.
Then, in Section~\ref{sec:algorithm}, we give a detailed description of the optimization approach, introduce the general scheme of \FASTATVO{} and analyze its global convergence. In Section~\ref{sec:experiments}, we report extensive numerical experiments on real--world data to show the performance of our new method and discuss the complexity of the algorithm in more detail. Finally, we draw some conclusions
in Section~\ref{conclusions}.

\section{Exact modularity relaxation via total variation}\label{sec:TV}
The total variation is classically defined for continuous functions: Given a real smooth function $u:X\to\R$ with $X\subseteq \R^n$, one defines $TV_X(u)=\int_{X}|\nabla u|$. The extension to the graph setting  $G=(V,E)$ and to graph valued functions $u:V\to\R$ (which we always implicitly identify with vectors in $\R^n$) is done by replacing the gradient with the nonlocal directional derivative $\triangledown u:E\to\R$, which measures the variation of $u$ along the edge $ij\in E$ via $\triangledown u(ij) = A_{ij}(u_i-u_j)$. The graph total variation of $u$ is then
\begin{equation}\label{eq:TV}
TV_G(u) = \frac 1 2 \sum_{ij}|\triangledown u(ij)| = \frac 1 2 \sum_{ij} A_{ij}|u_i-u_j|
\end{equation}
where, since the graph is undirected, the $1/2$ is required due to the fact that the nonzero entries of the adjacency matrix $A$ account for twice the number of edges.

In the graph setting, the total variation \eqref{eq:TV} is strictly related with the concept of Lov\'asz extension of a set valued function $F:2^V\to\R$, which we recall in the following
\begin{definition}
	Consider a  function $F:2^V \to \R$. Given a vector $x \in \R^n$, reorder it so that  $x_1\leq x_2\leq \ldots \leq x_n$ and let $C_i \subseteq V$ be the set
	$C_i=\{k \in V : x_k\geq x_i\}$.
	The Lov\'asz extension $f_F:\R^n \to \R$ of $F$ is defined by
	$$f_F(x) = \sum_{i=1}^{n-1} F(C_{i+1})(x_{i+1}-x_{i})+F(V)x_1 \,.$$
\end{definition}
The relation between total variation and Lov\'asz extension is through the set-valued cut function $K:2^V\to \R$ defined by $K(S) = \sum_{i\in S, j\notin S}A_{ij}$. In fact, we have
\begin{lemma}\label{lem:1}
For any graph $G$ and any $x\in \R^n$ it holds $f_K(x) = TV_G(x)$.
\end{lemma}
\begin{proof}
Assume for simplicity, and without loss of generality, that $x$ is such that $x_1\leq \dots \leq x_n$. Since by definition $K(V)=K(\emptyset) = 0$, we have
\begin{align*}
f_K(x) &= \sum_{i=1}^{n-1} K(C_{i+1})(x_{i+1}-x_{i})+K(V)x_1 \\
&=\sum_{i=1}^{n-1}|x_{i+1}-x_{i}|\sum_{k\in C_{i+1}, h\notin C_{i+1}} A_{hk}\\
&= \sum_{i=1}^{n-1}|x_{i+1}-x_{i}| \sum_{k\geq s(i+1), h< s(i+1)} A_{hk}\\
&=\sum_{k\geq h}A_{kh}|x_k-x_h|=TV_G(x)
\end{align*}
where, for $i\in V$, $s(i)\in V$ denotes the smallest index such that $x_k\geq x_{i}$ for all $k\geq s(i)$.
\end{proof}
The following result thus relates the total variation with the Lov\'asz extension of the modularity function.
\begin{lemma}\label{lem:2}
Let $G$ and $G_0$ be the graphs with adjacency matrices $A=(A_{ij})$ and $B=(d_id_j/\vol G)$, respectively. Then
$$
f_Q(x) =  \frac{TV_{G_0}(x)-TV_G(x)}{\vol G} = \frac 1 {2\vol G} \sum_{ij} \Big(\frac{d_id_j}{\vol G}-A_{ij}\Big) |x_i-x_j|
$$
for all $x\in \R^n$.
\end{lemma}
\begin{proof}
From the definition of $Q(S)$ we have
\begin{equation}\label{eq:TV1}
Q(S) = \frac{1}{\vol G}\Big(\sum_{ij\in S} A_{ij}-\frac{d_id_j}{\vol G}\Big) = \frac{1}{\vol G}\Big(\sum_{ij\in S}d_i - \frac{d_id_j}{\vol G} - (d_i - A_{ij})\Big)\, .
\end{equation}
Using the identities $d_i = \sum_{k\in V}A_{ik}=\sum_{k\in V}d_id_k/\vol G$, we get
\begin{align*}
\sum_{ij\in S}d_i-\frac{d_id_j}{\vol G} &= \sum_{ij\in S}\sum_{k\in V}\frac{d_id_k}{\vol G}-\frac{d_id_j}{\vol G}\\
& = \sum_{i\in S}\Big(\sum_{k\in V}\frac{d_id_k}{\vol G} - \sum_{j\in S}\frac{d_id_j}{\vol G}\Big) = \sum_{i\in S, j\notin S}\frac{d_id_j}{\vol G}
\end{align*}
and similarly
$$
\sum_{ij\in S}d_i-A_{ij} = \sum_{ij\in S}\sum_{k\in V}(A_{ik}-A_{ij})=\sum_{i\in S,j\notin S}A_{ij}\, .
$$
Thus, plugging the last two identities in \eqref{eq:TV1} we get
$$
Q(S) = \frac{1}{\vol G} \Big( K_0(S)-K(S)\Big)
$$
where $K_0(S)=\sum_{i\in S, j\notin S}B_{ij}$ is the cut function of the graph $G_0$ with adjacency matrix $B_{ij}=d_id_j/\vol G$.
From Lemma \ref{lem:1} and the linearity of the Lov\'asz extension,  $f_{F_1+F_2}=f_{F_1}+f_{F_2}$,  we conclude.
\end{proof}

So the Lov\'asz extension of $Q$ is the difference of two total variations or, equivalently, it is the total variation on the graph $G^{\pm}$ with real valued adjacency matrix $M=B-A = (d_id_j/\vol G- A_{ij})$, i.e.\ the graph whose nodes are $V$ and such that a weighted signed edge with weight $M_{ij}$ exists between $i$ and $j$ if and only if $M_{ij}\neq 0$. As this function is the main tool of our analysis we call it \textit{modularity total variation} and denote it from now on as
$$
TV_{Q}(x) = \frac 1 2 \sum_{ij} \Big(\frac{d_id_j}{\vol G}-A_{ij}\Big) |x_i-x_j|\, .
$$

Let $\one_S$ be the binary vector $(\one_S)_i=1$ if $i\in S$ and $(\one_S)_i=0$ otherwise. It is well known that, for any set valued function $F:2^V\to\R$, it holds $f_F(\one_S)=F(S)$ (see e.g.\ \cite{Bach2011}). Therefore, from Lemma \ref{lem:2}, $TV_{Q}(\one_S)/\vol G=f_Q(\one_S)=Q(S)$. However, for the specific modularity function, the following additional formula holds
\begin{lemma}\label{lem:interpol}
Let $S\subseteq V$ and let $a,b \in \R$. We have
$$
TV_{Q}(a \one_S + b \one_{\bar S}) = \vol G \cdot |a-b| \cdot Q(S) \,.
$$
\end{lemma}
\begin{proof}
Let  $M_{ij} = B_{ij}-A_{ij}= d_id_j/\vol G - A_{ij}$. As $|x_i-x_j|=0$ when $x=a\one_S+b\one_{\bar S}$ and $i,j$ are either both in $S$ or both in $\bar S$, we have
\begin{equation}\label{eq:1}
2\cdot TV_{Q}(a \one_S + b \one_{\bar S}) = \sum_{i\in S, j\notin S}M_{ij}|a-b| + \sum_{i\notin S, j\in S}M_{ij}|b-a| \,.
\end{equation}
Now notice that $\sum_{i\in S, j \notin S}M_{ij} = \sum_{i\in S,j\in V}M_{ij}-\sum_{i,j\in S}M_{ij}$, moreover
$$
\sum_{j\in V}M_{ij} = \frac{d_i}{\vol G}\sum_{j\in V}d_j - \sum_{j\in V} A_{ij} = d_i-d_i = 0\, .
$$
Combined with \eqref{eq:1} we get
$$
\frac{TV_{Q}(a \one_S + b \one_{\bar S})}{\vol G\, |a-b|} = \frac {1}{\vol G}\sum_{ij\in S}\Big(A_{ij}-\frac{d_id_j}{\vol G}\Big)= Q(S)\, ,
$$
concluding the proof.
\end{proof}

Lemma \ref{lem:interpol} shows that the modularity total variation $TV_Q$ somewhat interpolates the modularity set valued function $Q$ on vectors having only two different values across all the entries. Note that all these vectors identify uniquely a set of nodes and, in fact, there is a one-to-one correspondence between vectors of the type $a\one_S+b\one_{\bar S}$ and the subsets of $V$. However, a deeper and more relevant property of the modularity total variation $TV_Q$ is that its maximum coincides (up to scaling) with the maximum of $Q$. This was proved in \cite{tudisco2018community} and is the reason why $TV_Q$ is therein called a \textit{nonlinear exact relaxation} of $Q$. We generalize that result in the following theorem, where we also provide a simpler and self--contained proof. 
\begin{theorem}\label{thm:main}
Let $a, b >0$ and let $x^*\in \R^n$ be a solution of
\begin{equation}\label{eq:optimization_TVQ}
\begin{array}{ll}
\max& TV_Q(x)\\
\mbox{s.t. }&-a\leq x_i \leq b, \quad i=1,\dots, n.
\end{array}
\end{equation}
Then $TV_Q(x^*) =  \vol G \cdot (a+b)\cdot \max_{S\subseteq V} Q(S)$. 
\end{theorem}
\begin{proof}
For ease of  notation, and without loss of generality, we prove the theorem assuming $\vol G=1$. Let $\mathcal B(a, b)=\{x\in \R^n : -a\leq x_i\leq b\}$. As $u_S = b \one_S -a\one_{\bar S} \in \mathcal B(a,b)$, from Lemma \ref{lem:interpol} we get
$$
(a+b)\max_{S\subseteq V}Q(S) = \max_{S\subseteq V}TV_Q(u_S) \leq \max_{u \in \mathcal B(a,b)}TV_Q(u)\, .
$$
Vice-versa, recall the identities $TV_Q(x)=f_Q(x)$ and $Q(V)=0$. For any $u\in \mathcal B(a,b)$ assume w.l.o.g.\ $u_1 \le \ldots \le u_n$. We have
\begin{align*}
TV_Q(u) &= \sum_{i=1}^{n-1}Q(C_{i+1})(u_{i+1}-u_i) \leq \Big(\max_{i=1,\dots,n-1} Q(C_{i+1})\Big) (u_n-u_1)  \\
&\leq \Big(\max_{i=1,\dots,n-1} Q(C_{i+1})\Big) (a+b)\,
\end{align*}
where $C_{i} = C_i(u) =\{k : u_k\geq u_i\}$. Therefore
$$
\max_{u\in \mathcal B(a,b)}TV_Q(u) \leq (a+b)\max_{u\in \mathcal B(a,b)}\max_{i=1,\dots,n-1} Q(C_{i+1}) = (a+b)\max_{S\subseteq V}Q(S)
$$
which shows the reverse inequality. Thus the maximum of \eqref{eq:optimization_TVQ} coincides with $(a+b)$ times the maximum of $Q$, concluding the proof.
\end{proof}

By Theorem \ref{thm:main} we can now transform the leading module problem \eqref{eq:leading_module} into the box constrained optimization problem \eqref{eq:optimization_TVQ}, where we seek for the maximum of $TV_Q$ over the box $\mathcal B(a,b) = \{x\in \R^n:-a\leq x_i\leq b, \, i = 1,\ldots,n\}$. Note that we can freely choose the parameters $a,b>0$. Note moreover that the maximum $x^*$ is always attained on the border of $\mathcal B(a,b)$, i.e.\ it is of the form $x^*=b\one_S-a\one_{\bar S}$. This is shown by the following
\begin{corollary} \label{thm:sol_on_boundary}
Let $x^*$ be as in Theorem \ref{thm:main}. Then $x^* \in \partial \mathcal B(a,b)$, the  border of the box $\mathcal B(a,b)$.
\end{corollary}
\begin{proof}
Suppose by contradiction that $x^*\notin \partial \mathcal B(a, b)$. Then, there exist $a', b'>0$ with $-a<-a'<b'<b$ such that $x^*\in \mathcal B(a',b')$. From Theorem \ref{thm:main} we  have
$$(a+b)\max_{S\subseteq V}Q(S) = TV_Q(x^*) \leq \max_{x\in \mathcal B(a',b')}TV_Q(x) = (a'+b')\max_{S\subseteq V}Q(S)$$
and this is not possible as $(a+b)>(a'+b')$.
\end{proof}

\subsection{Comparison with the nonlinear spectral approach}\label{sec:grdca}
Nonlinear modularity eigenvectors are defined in \cite{tudisco2018community} as critical points of the Rayleigh quotient
$$
r_Q(x) =  \frac{TV_Q(x)}{\|x\|_\infty}
$$
and nonlinear eigenvectors $x^*$ corresponding to the largest eigenvalue  $\lambda^*=r_Q(x^*)$ are there used to locate leading communities. Due to the homogeneity of $TV_Q$, it is simple to observe that, when $a=b=1$, any maximizer of $r_Q$ is a solution to  \eqref{eq:optimization_TVQ} and vice-versa. In fact, note that for $a=b=1$, we have $B(a,b) = \{x:\|x\|_\infty\leq 1\}$. Thus, as $TV_Q(\lambda\, x)=\lambda\, TV_Q(x)$, for any scalar $\lambda>0$, we obtain
\begin{align*}
\max_{x\in \mathcal B(a,b)}TV_Q(x) &\geq \max_{\|x\|_\infty=1}TV_Q(x) = \max_{x\in\R^n}TV_Q\left(\frac{x}{\|x\|_\infty}\right) =\max_{x\in\R^n}r_Q(x)\\
& \geq \max_{x\in \mathcal B(a,b)} r_Q(x) \geq \max_{x\in \mathcal B(a,b)}TV_Q(x)
\end{align*}
which implies that all the above inequalities are indeed identities.

\section{Description of the nonlinear optimization approach}\label{sec:algorithm}
In order to compute an approximate solution to ~\eqref{eq:optimization_TVQ}, we first consider a smooth approximation of problem $TV_Q$. In particular, we replace
the modularity total variation with the following smooth function:
\begin{equation*}
TV_{Q}^p(x) = \frac 1 2 \sum_{ij} \Big(\frac{d_id_j}{\vol G}-A_{ij}\Big) |x_i-x_j|^p,
\end{equation*}
with $p>1$. Hence, the new problem we want to solve has the form:
\begin{equation}\label{eq:optimization_STVQ}
\begin{array}{ll}
\max& TV_Q^p(x)\\
\mbox{s.t. }& -a \le x_i \le b, \quad i=1,\dots, n,
\end{array}
\end{equation}
with $a, b > 0$.
Even if~\eqref{eq:optimization_STVQ} is still a hard non-convex problem, we can now use first-order methods to compute
a good approximate solution of the original problem~\eqref{eq:optimization_TVQ}.
In particular, in view of Corollary \ref{thm:sol_on_boundary}, we expect that also for the smooth problem~\eqref{eq:optimization_STVQ} good solutions lie on the boundary of the feasible set $\mathcal B(a,b)$.
For this reason, we develop an active-set method
that dynamically selects, at each iteration, a set of variables $x_i$ to be safely fixed at the boundary, and then optimizes over the free subspace (i.e., the subspace of variables not fixed).
Further, we propose a global optimization framework, which embeds the active-set algorithm, to improve the overall quality function score.

As we will see, even though $TV_Q^p$ and $TV_Q$ may differ when $p>1$, using such a  smooth approximation and the proposed optimization algorithm allows us to both
\begin{itemize}
\item  significantly increase the value of $TV_Q$ (or, equivalently, of $Q$), and
\item significantly reduce the required computational time
\end{itemize}
with respect to the previously proposed Generalized RatioDCA.


In the following, we first describe our active-set algorithm (in Subsection~\ref{sub:LocalSearch})
and then we report the global optimization framework (in Subsection~\ref{sub:Global}).
Finally, in Subsection~\ref{sub:complexity} we carry out a complexity analysis of function and gradient computation,
which is useful to design an efficient implementation of the algorithm.

\subsection{An algorithm for the approximate total variation optimization}\label{sub:LocalSearch}
In order to describe the algorithm we use for solving problem~\eqref{eq:optimization_STVQ},
let us consider a general bound-constrained optimization problem of the following form:
\begin{equation}\label{prob}
\min\ \{f(x): l \le x \le u\},
\end{equation}
where the inequalities between vectors are meant to hold componentwise,
$f \in C^1(\R^n)$, $l,u \in \R^n$ and $l<u$. Clearly, problem~\eqref{eq:optimization_STVQ} is a particular case of~\eqref{prob},
obtained by setting $f(x) = -TV_Q^p(x)$ and $l_i = -a$, $u_i = b$, $i = 1,\ldots,n$.

From now on, given a vector $x \in \R^n$, let us indicate by $[x]^\sharp$ the projection of $x$ onto the feasible set of problem~\eqref{prob}.
Since we consider a non-convex objective function $f$, we also need to recall the definition of stationarity:
a point $x^* \in [l,u]$ is said to be \textit{stationary} for problem~\eqref{prob} if
\begin{equation*}
\norm{x^* - [x^* - \nabla f(x^*)]^{\sharp}} = 0,
\end{equation*}
or equivalently,
\begin{subequations}\label{stationarity}
\begin{align}
\nabla_i f(x^*) = 0,   & \quad i \colon l_i < x^*_i < u_i, \\
\nabla_i f(x^*) \ge 0, & \quad i \colon x^*_i = l_i, \\
\nabla_i f(x^*) \le 0, & \quad i \colon x^*_i = u_i.
\end{align}
\end{subequations}

The algorithm we propose here belongs to the class of \textit{active-set methods} (see, e.g.,
\cite{cristofari:2017,curtis2015globally,desantis:2012,ferreau2014qpoases,hager:2006,nocedal2006sequential} and references therein),
where an estimate of the active (or binding) inequality constraints at the final solution is iteratively updated at each iteration.
In box-constrained problems, the estimate gets in practice a subset of variables that can be fixed to the bounds at each iteration.
This reduces the complexity of the search direction step and it turns out to be extremely useful in our context,
where good points are likely to lie on the boundary of the feasible set.

Using the stationarity conditions~\eqref{stationarity}, we can define the following
active and non-active set estimates:
\begin{subequations}\label{active_set_estimates}
\begin{align}
A_l(x) & = \{i \colon x_i = l_i, \, \nabla_i f(x)>0\}, \label{Al} \\
A_u(x) & = \{i \colon x_i = u_i, \, \nabla_i f(x)<0\}, \label{Au} \\
N(x)   & = \{i \colon  i \notin A_l(x) \cup A_u(x)\}. \label{N}
\end{align}
\end{subequations}
In particular, for every feasible point $x$, the sets $A_l(x)$ and $A_u(x)$ contain the indices of those variables we estimate to be active
at the lower and the upper bounds, respectively, in the final solution returned by the algorithm.
Vice--versa, $N(x)$ contains the indices of those variables we estimate to be non-active at the final solution.
Note that, by a slight abuse of terminology, we say that  a variable is \textit{active} if one of its bound constraints is active.

By adapting classic results for active-set strategies (see, e.g.,~\cite{dipillo:1984,facchinei:1995}), it is easy to prove the following proposition, whose proof is omitted here for the sake of brevity.
\begin{proposition}\label{prop:estim}
Let $x^*$ be a stationary point for problem~\eqref{prob}. Then, there exists a neighborhood ${\Omega}(x^*,\rho) = \{x: \|x-x^*\|\leq \rho\}$ such that
\begin{gather*}
\{i \, \colon \, x^*_i = l_i, \, \nabla_i f(x^*)>0 \} \subseteq A_l(x) \subseteq \{i \, \colon \, x^*_i = l_i\}, \\
\{i \, \colon \, x^*_i = u_i, \, \nabla_i f(x^*)<0 \} \subseteq A_u(x) \subseteq \{i \, \colon \, x^*_i = u_i\},
\end{gather*}
for each $x \in {\Omega}(x^*,\rho)$.
\end{proposition}

The detailed description of the proposed algorithm, that we name Fast Active-SeT based Approximate Total Variation Optimization (\FASTATVO),
can be found in Appendix~\ref{appendix}, together with its convergence analysis.
In what follows, instead, we  detail the main ideas behind \FASTATVO{} and a sketch of its  structure. 
First, let us emphasize the two features that, together with the active-set estimate given in~\eqref{active_set_estimates}, most characterize our algorithm:
\begin{enumerate}[label=(\roman*)]
\item the use of a non-monotone stabilization technique \cite{grippo1991class} that allows us to avoid the objective function computation at every iteration;
\item the possibility to update only a subset of variables at every iteration.
\end{enumerate}
Both the above algorithmic features are particularly well suited for our specific problem, since, as we will show in Subsection~\ref{sub:complexity},
the computation of $TV_Q^p$
\begin{itemize}
\item is in general computationally expensive (and then, it is convenient to avoid computing it at every iteration);
\item can be implemented in an efficient way when a small subset of variables is changed from one iteration to another
    (and then, decomposition approaches are convenient, especially when the problem dimension is large).
\end{itemize}
In particular, note that computing $TV_Q^p$ on a vector $x\in \mathbb R^n$ can be significantly more expensive than computing the modularity $Q$ of a set $S\subseteq V$, especially for small sets $S$, as the two would be equivalent only if the number of nonzero entries of $x$ is exactly $|S|$, which we cannot expect in general.

The main steps of \FASTATVO\ are overviewed in Algorithm~\ref{alg:FAST_ATVO_SHORT}
(see Algorithm~\ref{alg:FAST_ATVO} in Appendix~\ref{appendix} for a detailed description).

\begin{algorithm}[h]
\caption{\FASTATVO($x^0$) -- short scheme}
\label{alg:FAST_ATVO_SHORT}
\begin{algorithmic}\setlength{\itemsep}{2pt}
\scriptsize
\vspace{5pt}
\item[] Given a feasible point $x^0$, set $k = 0$
\item[] While $x^k$ is non-stationary for problem~\eqref{prob}
\item[]\hspace{0.5truecm} Compute $A_l^k=A_l(x^k)$, $A_u^k=A_u(x^k)$ and $N^k=N(x^k)$
\item[]\hspace{0.5truecm} Choose $W^k \subseteq N^k$, set $d^k_{A_l^k} = 0$, $d^k_{A_u^k} = 0$, $d^k_{N^k \setminus W^k} = 0$ and compute $d^k_{W^k}$
\item[]\hspace{0.5truecm} Compute a stepsize $\alpha^k$ by a non-monotone stabilization strategy, set
                           \par $x^{k+1} = [x^k+\alpha^k d^k]^{\sharp}$ and $k = k + 1$
\item[] End while
\vspace{5pt}
\end{algorithmic}
\end{algorithm}

We see that,
at the beginning of every iteration $k$, we compute the active and non-active set estimates as in~\eqref{active_set_estimates} and define
\begin{equation*}
A_l^k = A_l(x^k), \quad A_u^k = A_u(x^k) \quad \text{and} \quad N^k = N(x^k).
\end{equation*}
Then, a (non-empty) working set $W^k \subseteq N^k$ is chosen according to a Gauss-Southwell-type (or greedy) rule and a first-order
search direction $d^k$ is computed such that $d^k_{A_l^k} = 0$, $d^k_{A_u^k} = 0$ and $d^k_{N^k \setminus W^k} = 0$.
Namely, $d^k$ is computed in the reduced variable subspace defined by $W^k$.
Finally, a non-monotone stabilization strategy is used to compute a stepsize $\alpha^k$ and generate the new iterate $x^{k+1} = [x^k+\alpha^k d^k]^{\sharp}$.

The following theorem shows the global convergence of \FASTATVO\ to stationary points.
\begin{theorem}\label{thm:conv}
Let $\{x^k\}$ be the sequence of points generated by \FASTATVO. Then, either
an  integer $\bar k\geq 0$ exists  such  that $x^{\bar k}$
is  a  stationary  point  for  problem \eqref{prob}, or else the sequence
$\{x^k\}$ is infinite and every limit point $x^*$
of the sequence is a stationary point for problem \eqref{prob}.
\end{theorem}
\begin{proof}
See Appendix~\ref{appendix}.
\end{proof}

\subsection{A global optimization strategy for improving the modularity}\label{sub:Global}
Even though \FASTATVO{} is able to exploit the structure of the modularity total variation optimization problem \eqref{eq:optimization_STVQ},
we can only guarantee convergence to a stationary point. This is not surprising as optimizing the network modularity $Q$ is in general strongly NP-complete \cite{brandes2007finding} and thus, due to Theorem \ref{thm:main}, we cannot expect convergence to global minimizers of \eqref{eq:optimization_STVQ}. 
Although, in practice, stationary points obtained with \FASTATVO{} often provide quality results, we can employ a global optimization
strategy in order to further improve the quality of the final solution. Of course, this has an additional computational cost.

The arguably most elementary global optimization strategy one can use is \textit{Multistart},
which consists of repeatedly running a local optimization algorithm from some randomly chosen starting points and finally
picking the best one. This is the strategy that is used in \cite{tudisco2018community}.
In our case, a few experiments showed that Multistart is quite inefficient. Therefore, in order to get a better graph partition
without significantly increasing the CPU time, we propose here a slightly more sophisticated technique.
In particular, we adopted a form of iterated local search/basin hopping strategy (see, e.g., \cite{grosso2007population,leary2000global})
that we call Partition \& Swap (\texttt{PS}). The scheme of \texttt{PS} is reported in Algorithm~\ref{BH}.

\begin{algorithm}[h]
\caption{\texttt{PS}($x^0$)}
\label{BH}
\begin{algorithmic}\setlength{\itemsep}{2pt}
\scriptsize
\vspace{5pt}
\item[]$\,\,\,0$ Given a feasible point $x^0$
\item[]$\,\,\,1$ Set $\bar x^0 = \FASTATVO(x^0)$ and $k = 1$
\item[]$\,\,\,2$ For $k=1,2,\ldots$
\item[]$\,\,\,3$\hspace{15pt} Set $x^k= \SWAP(\bar x^{k-1})$
\item[]$\,\,\,4$\hspace{15pt} Set $y^k = \FASTATVO(x^k)$
\item[]$\,\,\,5$\hspace{15pt} If $TV_Q(y^k) < TV_Q(\bar x^{k-1})$ 
\item[]$\,\,\,6$\hspace{30pt} Set $\bar x^k = y^k$
\item[]$\,\,\,7$\hspace{15pt} Else
\item[]$\,\,\,8$\hspace{30pt} Set $\bar x^k =\bar x^{k-1}$
\item[]$\,\,\,9$\hspace{15pt} End if
\item[]$10$ End for
\vspace{5pt}
\end{algorithmic}
\end{algorithm}

We see that \texttt{PS} consists in applying, at each iteration, a perturbation to the current stationary point
(according to the \SWAP\ strategy described later) and starting \FASTATVO\ from the new perturbed point.
When \FASTATVO\ gets a better solution, in terms of modularity total variation, then $\bar x_k$ (the current best solution) is updated;
otherwise, $\bar x^k$ is left unchanged and the procedure is repeated until a maximum number
of iterations is reached.

For what concerns the \SWAP\ strategy, it is based on a simple idea of shifting a given percentage of variables in $\bar x^k$
to the bound with the opposite sign.
We illustrate the \SWAP{} scheme in Algorithm~\ref{alg:SWAP}.

\begin{algorithm}[h]
\caption{\texttt{SWAP}($x$)}
\label{alg:SWAP}
\begin{algorithmic}\setlength{\itemsep}{2pt}
\scriptsize
\vspace{5pt}
\item[]$0$ Given a feasible point $x$, choose $\sigma \in [0,100]$
\item[]$1$ Partition $\{1,\ldots,n\}$ in $I_l$ and $I_u$, such that $x_i \le 0$ if $i \in I_l$ and $x_i \ge 0$ if $i \in I_u$
\item[]$2$ Set $y = x$
\item[]$3$ Randomly pick $\sigma\%$ of indices from $I_l$ and set the corresponding components $y_i$ to $u_i$
\item[]$4$ Randomly pick $\sigma\%$ of indices from $I_u$ and set the corresponding components $y_i$ to $l_i$
\item[]$5$ Return $y$
\vspace{5pt}
\end{algorithmic}
\end{algorithm}

\subsection{Complexity analysis of function and gradient computation}\label{sub:complexity}
When solving problem~\eqref{prob}, a severe bottleneck for computational efficiency
is represented by the time spent for computing the objective function and its gradient,
since, as to be shown below, both of them have a cost that may grow quadratically with the problem dimension.
Therefore, a complexity analysis of these computations needs to be carry out
in order to understand the most efficient way to perform the above operations in practice.

Let $f(x)$ be the objective function of problem~\eqref{prob}.
Denoting by $M$ the symmetric matrix such that $M_{ij} = d_i d_j / \vol G - A_{ij}$ for all $i,j = 1,\ldots$,
we can express the objective function $f(x) = TV_Q^p(x)$ as
\begin{equation}\label{obj}
f(x) = \sum_{i=1}^{n-1} \sum_{j=i+1}^n M_{ij} \abs{x_i-x_j}^p
\end{equation}
and its gradient as
\begin{equation}\label{grad}
\nabla_i f(x) = p \sum_{\substack{j=1 \\ j \ne i}}^n M_{ij} \, \sign{(x_i-x_j)} \abs{x_i-x_j}^{p-1}, \quad i = 1,\ldots,n.
\end{equation}
We see that every function evaluation and every gradient evaluation have a cost
(in terms of number of arithmetic operations)
\begin{equation}\label{cost_grad}
\mathcal O\biggl(\dfrac{n (n-1)}2\biggr),
\end{equation}
respectively.
In our implementation of \FASTATVO, we use some tricks to reduce these costs by exploiting the structure of $f$ and $\nabla f$.
For convenience of exposition, we first focus on the gradient computation.

Given any iterate $x^k$, assume that the working set $W^k$ has been chosen and the successive iterate $x^{k+1}$ has been produced
(as $x^{k+1} = x^k + \alpha^k d^k$, with $\alpha^k > 0$ and $d^k_i = 0$ for all $i \notin W^k$).
For any index $i = 1,\ldots,n$, let us define the two functions
\begin{gather*}
\phi_i(x) = p \sum_{\substack{j=1 \\ j \ne i \\ j \in W^k}}^n M_{ij} \, \sign{(x_i-x_j)} \abs{x_i-x_j}^{p-1}, \\
\varrho_i(x) = p \sum_{\substack{j=1 \\ j \ne i \\ j \notin W^k}}^n M_{ij} \, \sign{(x_i-x_j)} \abs{x_i-x_j}^{p-1}.
\end{gather*}
We see that $\phi_i(x)$ and $\varrho_i(x)$ differ only in that the summation is over $j \in W^k$ in the first function
and over $j \notin W^k$ in the second function.
Using~\eqref{grad},
clearly we have $\nabla_i f(x) = \phi_i(x) + \varrho_i(x)$ for all $i = 1,\ldots,n$.
Now, consider any index $h \notin W^k$. Since $x^{k+1}_h = x^k_h$, we can write
\begin{equation*}
\nabla_h f(x^{k+1}) = \phi_h(x^{k+1}) + \varrho_h(x^k) = \phi_h(x^{k+1}) + \nabla_h f(x^k) - \phi_h(x^k).
\end{equation*}
In other words, we can obtain $\nabla_h f(x^{k+1})$ by computing $\phi_h(x^{k+1})$ and $\phi_h(x^k)$
(since $\nabla_h f(x^k)$ is known at the currrent iteration).
So, we can get $\nabla f(x^{k+1})$ as follows:
\begin{equation}\label{grad_fast}
\nabla_i f(x^{k+1}) =
\begin{cases}
\phi_i(x^{k+1}) + \varrho_i(x^{k+1}), \quad               & i \in W^k \\
\phi_i(x^{k+1}) + \nabla_i f(x^k) - \phi_i(x^k), \quad & i \notin W^k.
\end{cases}
\end{equation}

Now, we are interested in the cost of computing $\nabla f$ by~\eqref{grad_fast} in order to compare this cost with the one needed by~\eqref{grad}.
First, we see that~\eqref{grad_fast} requires to compute $\phi_i(x^{k+1})$ for all $i \in W^k$.
The cost of this operation is the order of the number of pairs of indices belonging to $W^k$, that is
\begin{equation*}
\mathcal O\biggl(\dfrac{\abs{W^k} (\abs{W^k}-1)}2\biggr).
\end{equation*}
Then, we have the cost of computing $\varrho_i(x^{k+1})$ for all $i \in W^k$ and $\phi_i(x^{k+1})$ for all $i \notin W^k$,
which is the order of the number of pairs of indices such that exactly one of them belongs to $W^k$.
The number of these pairs is equal to $\abs{W^k}(n-\abs{W^k})$, obtained as
$n(n-1)/2$ (total number of pairs of indices) minus $\abs{W^k}(\abs{W^k}-1)/2$ (number of pairs of indices belonging to $W^k$)
minus $(n-\abs{W^k})(n-\abs{W^k}-1)/2$ (number of pairs of indices not belonging to $W^k$).
Therefore, the cost of computing $\varrho_i(x^{k+1})$ for all $i \in W^k$ and $\phi_i(x^{k+1})$ for all $i \notin W^k$ is
\begin{equation*}
\mathcal O\bigl(\abs{W^k}(n-\abs{W^k})\bigr).
\end{equation*}
The last cost in~\eqref{grad_fast} is the one for computing $\phi_i(x^k)$ for all $i \notin W^k$. This cost is still
the order of the number of pairs of indices such that exactly one of them belongs to $W^k$, that is,
$\mathcal O\bigl(\abs{W^k}(n-\abs{W^k})\bigr)$.

Summing all up, we get that computing $\nabla f$ as in~\eqref{grad_fast} has a cost of the order of
$\abs{W^k} (\abs{W^k}-1)/2 + 2\abs{W^k}(n-\abs{W^k})$, that is,
\begin{equation}\label{cost_grad_fast}
\mathcal O\biggl(\dfrac{\abs{W^k}(4n-3\abs{W^k}-1)}2\biggr).
\end{equation}

Comparing~\eqref{cost_grad_fast} with~\eqref{cost_grad}, we can conclude that applying~\eqref{grad_fast} is more convenient than~\eqref{grad} if
$\abs{W^k}(4n-3\abs{W^k}-1) < n (n-1)$. It is straightforward to verify that this inequality is satisfied if
\begin{equation*}
\abs{W^k} < \frac{n-1}3.
\end{equation*}
So, in our implementation of \FASTATVO, at each iteration we use either~\eqref{grad} or~\eqref{grad_fast} to compute $\nabla f(x^k)$,
depending on the dimension of $\abs{W^k}$.
Let us also point out that this complexity analysis shows that a decomposition approach, such as the one used in \FASTATVO,
can be particularly well suited for solving problem~\eqref{prob}.

For what concerns the computation of the objective function, we may use similar arguments as above to show that a cost lower than $\mathcal O(n(n-1)/2)$
can be obtained if $\abs{W^k}$ is sufficiently small.
Though, in our implementation of \FASTATVO\ we use the following formula to compute the objective function:
\begin{equation}\label{f_fast}
f(x) = \frac{\nabla f(x)^T x}p.
\end{equation}
Since we need to compute $\nabla f(x^k)$ at every iteration of \FASTATVO, this choice seems particularly convenient as the cost of~\eqref{f_fast}
is linear once the gradient is known.
The only drawback is that, at a given iteration $k$, we may compute the objective function many times before producing the successive $x^{k+1}$:
in this case we would compute many gradients that are not needed.
But this can only occur if the unit stepsize is not accepted in the line search procedure (see Algorithm~\ref{alg:FAST_ATVO} in Appendix~\ref{appendix}),
and our experiments showed that this is an unlikely event.

\section{Experiments}\label{sec:experiments}
In this section we apply our method to several real-world and established random benchmark networks with the aim of highlighting the improvements that the exact modularity total variation approach here proposed ensures over the standard linear method  \cite{newman2006finding} and the  nonlinear eigenvector method previously proposed in \cite{tudisco2018community}.

We subdivide the discussion in two parts:
first,
in Subsection \ref{sec:implementation_details} we provide details on the parameters used in our implementation of \FASTATVO{}; then, in Subsection \ref{sec:results} we report extensive numerical results on sixteen datasets.

\subsection{Implementation details}\label{sec:implementation_details}
In this subsection, we provide details on the parameters used in our experiments to approximate problem~\eqref{eq:optimization_STVQ}
and to execute \FASTATVO.

For what concerns problem~\eqref{eq:optimization_STVQ}, we used $p = 1.4$ to define the function $TV_Q^p(x)$. This choice of $p$ has been guided by an extensive parameter--tuning phase and it has been chosen because it performed best overall on the datasets we have analyzed. As for the choice of the box constraints, we set $a=b=1$ to have a more fair comparison with the Generalized RatioDCA approach, which was designed to solve that particular problem setting.

For what concerns~\FASTATVO, with reference to Algorithm~\ref{alg:FAST_ATVO} reported in Appendix~\ref{appendix}, we set
$Z = 20$, $M = 100$, $\Delta_0 = 1\text e20$, $\beta = 0.99$, $\delta = 0.5$ and $\gamma = 1$e$-3$.
The working set $W^k$ is computed randomly at every iteration (except for the index $\hat \imath^k$ that is always included in $W^k$)
and its dimension starts from $2$ and is gradually increased through the iterations up to $\max(10,\min(1000,0.03 n))$.
The search direction is computed as explained in Appendix~\ref{appendix}, using $\mu_\text{min} = 1$e$-10$ and $\mu_\text{max} = 1$e$10$.

Moreover, we use the following initialization strategy, that our preliminary experiments suggested to be useful for improving performances in practice:
given any starting point $x^0$, all negative (positive) components of $x^0$ are set to the lower (upper) bound.

Finally, in the global optimization strategy described in Subsection~\ref{sub:Global} and Algorithm \ref{alg:SWAP} we set $\sigma = 75$, which was found to give good results.

All the parameters used in the experiments have been chosen after a parameter--tuning phase performed over the analyzed data sets.
So, in our experience, these parameters can be considered a good general setting.

\subsection{Results}\label{sec:results}
In this section we apply our method to several real-world and random benchmark networks with the aim of highlighting: (a) the improvements that \FASTATVO{} ensures with respect to the previously proposed Generalized RatioDCA and (b) the overall improvement that the continuous modularity total variation approach we propose here ensures over the standard linear relaxation approach.

In particular, in our experiments we first tested the local optimization method \FASTATVO, described in Subsection~\ref{sub:LocalSearch},
and then the global optimization method \texttt{PS}, described in Subsection~\ref{sub:Global}.
We considered more than $30$ datasets, but, for the sake of brevity, here we report results for $16$ moderate-to-large size networks, ranging from around $7000$ up to around $65000$ nodes. The results we obained on the remaining datasets are aligned with the ones presented here. The list of the $16$ datasets is shown in Table \ref{tab:datasets}.
All the datasets are publicly available and all the methods are implemented in \texttt{C++}. The software for implementing \FASTATVO{} is available at
\begin{center}
\tt \url{https://github.com/acristofari/fast-atvo}
\end{center}

\begin{table}[t]
\centering
\caption{List of datasets used in the experiments, with corresponding sizes (in terms of number of nodes) and references. }
{\scriptsize
\begin{tabular}{|l|l|l|l|}
\hline
\textbf{Dataset ID} & \textbf{Dataset name} & \textbf{\# nodes} & \textbf{Reference}\bigstrut[t] \bigstrut[b] \\
\hline
   1 & \verb!geom!                            & 7343  & \cite{pajek_datasets} \bigstrut[t] \\
   2 & \verb!as-735!                          & 7716  & \cite{as_datasets}\\
   3 & \verb!ca-HepTh!                        & 9877  & \cite{ca_datasets}\\
   4 & \verb!vsp_c-30_data_data!              & 11023 & \cite{dimacs_datasets}\\
   5 & \verb!Oregon-1!                        & 11492 & \cite{as_datasets}\\
   6 & \verb!ca-HepPh!                        & 12008 & \cite{pajek_datasets}\\
   7 & \verb!vsp_befref_fxm_2_4_air02!        & 14109 & \cite{dimacs_datasets}\\
   8 & \verb!ca-AstroPh!                      & 18772 & \cite{ca_datasets}\\
   9 & \verb!ca-CondMat!                      & 23133 & \cite{ca_datasets}\\
  10 & \verb!rgg_n_2_15_s0!                   & 32768 & \cite{dimacs_datasets}\\
  11 & \verb!vsp_sctap1-2b_and_seymourl!      & 40174 & \cite{dimacs_datasets}\\
  12 & \verb!vsp_model1_crew1_cr42_south31!   & 45101 & \cite{dimacs_datasets}\\
  13 & \verb!Words28!                          & 52652 & \cite{pajek_datasets}\\
  14 & \verb!vsp_bump2_e18_aa01_model1_crew1! & 56438 & \cite{dimacs_datasets}\\
  15 & \verb!loc-Brightkite!                  & 58228 & \cite{brightkite_datasets}\\
  16 & \verb!rgg_n_2_16_s0!                   & 65536 & \cite{dimacs_datasets}\bigstrut[b] \\
\hline
\end{tabular}
}
\label{tab:datasets}
\end{table}

The networks we selected  are primarily real networks, plus several randomly generated graphs borrowed from established benchmarks.
In particular: \verb!geom! is a collaboration network in computational geometry \cite{pajek_datasets}; \verb!as-735! and \verb!Oregon-1! are snapshots of Internet at the level of autonomous systems \cite{as_datasets}; \verb!ca-HepTh!, \verb!ca-HepPh!, \verb!ca-AstroPh! and \verb!ca-CondMat! are collaboration networks based on arXiv papers in High Energy Physics, Astro Physics and Condensed matter, respectively \cite{ca_datasets};
\verb!rgg_n_2_15_s0! and \verb!rgg_n_2_16_s0! are random geometric graph with $2^{15}$ (resp.\ $2^{16}$) vertices used for the DIMACS10 challenge on graph clustering and partitioning \cite{dimacs_datasets};
 \verb!Words28! is a dictionary graph built using a collection of close words in modern English \cite{pajek_datasets}; \verb!vsp_c-30_data_data!, \verb!vsp_befref_fxm_2_4_air02!,  \verb!vsp_sctap1-2b_and_seymourl!, \verb!vsp_model1_crew1_cr42_south31! and \verb!vsp_bump2_e18_aa01_model1_crew1! are DIMACS10 star--like random benchmark graphs consisting of a mixture of  social networks, finite-element graphs, VLSI chips, peer-to-peer networks and graphs from optimization solvers \cite{dimacs_datasets}; \verb!loc-Brightkite! is a
Location-based social network \cite{brightkite_datasets}.

Performance results are shown in Tables \ref{tab:results_local_1} and \ref{tab:results_local_2} and Figure \ref{fig:boxplot_local}.
Table \ref{tab:results_local_1} compares the modularity value obtained with \FASTATVO{}, with Generalized RatioDCA and with the standard linear spectral method. The results shown in this table have been obtained using the leading eigenvector of the modularity matrix (what we call \textit{linear eigenvector}) as a starting point of both \FASTATVO{} and Generalized RatioDCA algorithms. From this table we can see that \FASTATVO{} generally performs best with improvements in terms of modularity value up to 140\% more than the linear eigenvector (for \verb!ca-HepTh! dataset)  and up to 97\% more than the Generalized RatioDCA method (for \verb!vsp_befref_fxm_2_4_air02! dataset). Moreover, note that the communities identified by each method significantly different. This is also highlighted in Table \ref{tab:results_local_1}, where we compare the sizes of such communities.

\begin{table}[t!]
\centering
\caption{Performance comparison using the linear eigenvector as starting point. The values shown here correspond to the value of $Q(S^*)$
and the size of $S^*$, where  $S^*$ is the community identified by optimal thresholding the output of the three methods:
linear (\textbf{Q$_\text{linear}$} and \textbf{SIZE$_\text{linear}$}, respectively),
Generalized RatioDCA ($\textbf{Q}_\textbf{R}$ and \textbf{SIZE$_\text{R}$}, respectively) and
\FASTATVO{} ($\textbf{Q}_\textbf{F}$ and \textbf{SIZE$_\text{F}$}, respectively).
In the columns corresponding to the size of $S^*$ we also show, between round brackets, the percentage of nodes in $S^*$ as compared with the whole set of nodes.}
{\scriptsize
{\begin{tabular}{| c | c c c | c c c | c c |}
\hline
\textbf{Dataset} \bigstrut[t]
& \multirow{2}*{\textbf{Q$_\text{linear}$}}
& \multirow{2}*{\textbf{Q$_\text{R}$}}
& \multirow{2}*{\textbf{Q$_\text{F}$}}
& \multirow{2}*{\textbf{SIZE$_\text{linear}$}}
& \multirow{2}*{\textbf{SIZE$_\text{R}$}}
& \multirow{2}*{\textbf{SIZE$_\text{F}$}}
& \multirow{2}*{$\dfrac{\textbf Q_\textbf{F}}{\textbf Q_\textbf{linear}}$}
& \multirow{2}*{$\dfrac{\textbf Q_\textbf{F}}{\textbf Q_\textbf{R}}$} \bigstrut[b] \\
\textbf{ID} & & & & & & & &  \bigstrut[b] \\
\hline
1 &  0.28 &  0.29 &  0.40 & 3522 (48\%) & 3128 (43\%) & 2467 (34\%) &  1.46 &  1.40 \bigstrut[t] \\
2 &  0.20 &  0.25 &  0.37 & 3072 (40\%) & 3530 (46\%) & 3187 (41\%) &  1.89 &  1.51 \\
3 &  0.17 &  0.39 &  0.40 & 1517 (15\%) & 3624 (37\%) & 3324 (34\%) &  2.40 &  1.03 \\
4 &  0.33 &  0.41 &  0.47 & 5285 (48\%) & 3673 (33\%) & 4920 (45\%) &  1.42 &  1.17 \\
5 &  0.28 &  0.31 &  0.39 & 4448 (39\%) & 5051 (44\%) & 5494 (48\%) &  1.38 &  1.25 \\
6 &  0.35 &  0.36 &  0.41 & 1117 (9\%) & 971 (8\%) & 1166 (10\%) &  1.16 &  1.15 \\
7 &  0.21 &  0.21 &  0.42 & 3346 (24\%) & 3346 (24\%) & 6738 (48\%) &  1.97 &  1.97 \\
8 &  0.25 &  0.33 &  0.36 & 4186 (22\%) & 7448 (40\%) & 6334 (34\%) &  1.44 &  1.07 \\
9 &  0.23 &  0.38 &  0.39 & 8105 (35\%) & 10801 (47\%) & 6307 (27\%) &  1.68 &  1.04 \\
10 &  0.32 &  0.44 &  0.50 & 11613 (35\%) & 14228 (43\%) & 16333 (50\%) &  1.55 &  1.12 \\
11 &  0.23 &  0.23 &  0.39 & 16754 (42\%) & 16754 (42\%) & 6314 (16\%) &  1.70 &  1.70 \\
12 &  0.26 &  0.26 &  0.40 & 6176 (14\%) & 6176 (14\%) & 13438 (30\%) &  1.56 &  1.56 \\
13 &  0.32 &  0.46 &  0.43 & 7565 (14\%) & 11538 (22\%) & 23620 (45\%) &  1.32 &  0.93 \\
14 &  0.27 &  0.35 &  0.44 & 8289 (15\%) & 9013 (16\%) & 22045 (39\%) &  1.66 &  1.25 \\
15 &  0.20 &  0.35 &  0.37 & 4495 (8\%) & 24931 (43\%) & 18081 (31\%) &  1.83 &  1.07 \\
16 &  0.31 &  0.44 &  0.50 & 23054 (35\%) & 28106 (43\%) & 32600 (50\%) &  1.60 &  1.14 \bigstrut[b] \\
\hline
\end{tabular}}
}
\label{tab:results_local_1}
\end{table}

The stochastic nature of the working set selection combined with the non-monotone stabilization strategy we use to generate the new iterates are major differences between the proposed \FASTATVO{} and the Generalized RatioDCA method. In our opinion, these two features allow the method to escape from ``bad'' local optima and quickly find much better modules. In fact, Generalized RatioDCA seems to often get trapped into a local optima somewhat near the initial linear eigenvector assignment, while \FASTATVO{} is able to efficiently move further away. The phenomenon is particularly evident for the datasets $\{1,6,7,11,12\}$, as shown in Table \ref{tab:results_local_1}, and it is further highlighted in the example drawing of Figure \ref{fig:community_example}.

\begin{figure}[t!]
\centering
\subfloat[Boxplots for graphs with less than 20000 nodes.]
{\includegraphics[scale=0.6]{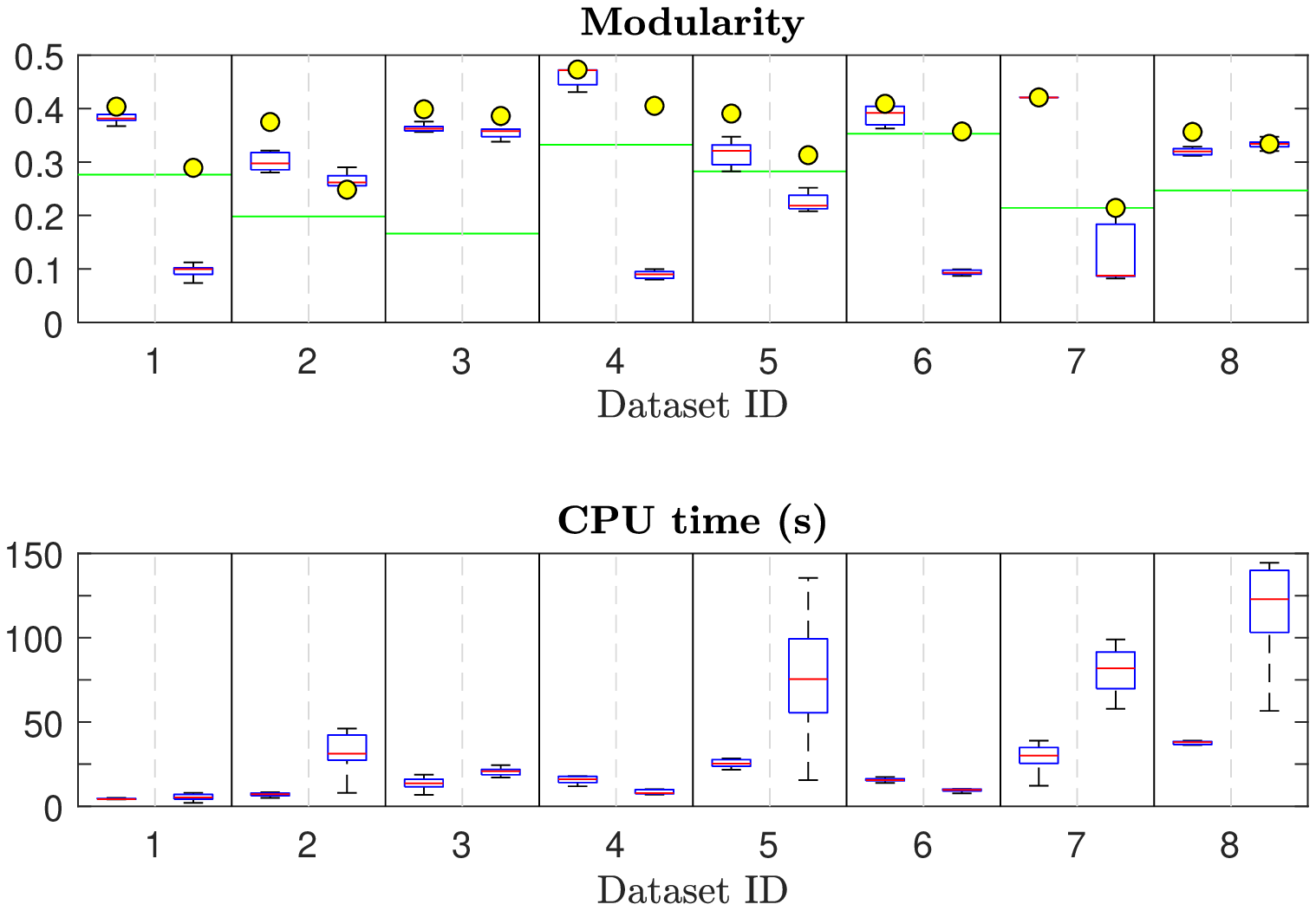}\label{subfig:boxplot_local_1}}

\subfloat[Boxplots for graphs with more than 20000 nodes.]
{\includegraphics[scale=0.6]{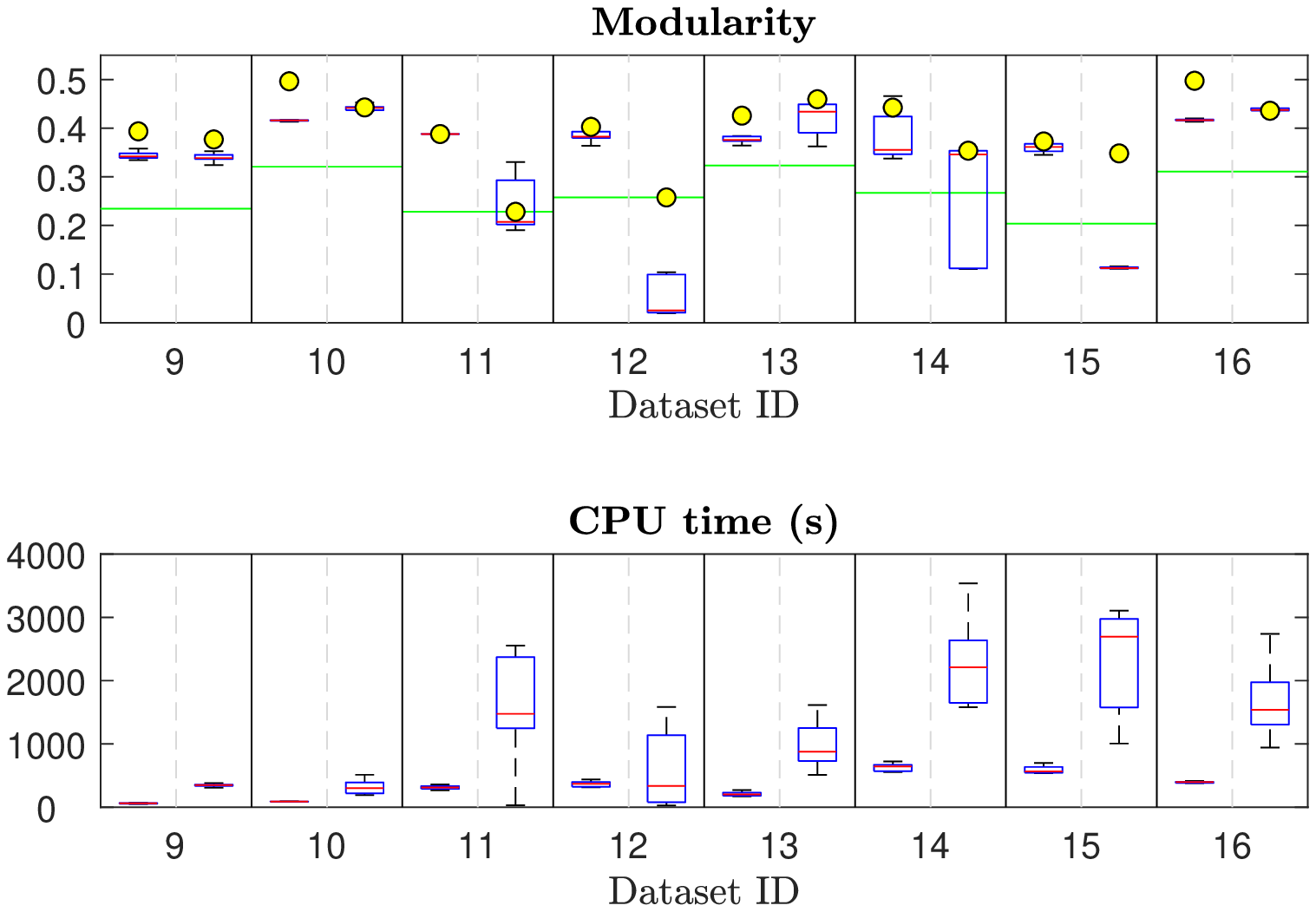}\label{subfig:boxplot_local_2}}
\caption{Boxplots of modularity values (upper panels) and execution times (lower panels) for \FASTATVO{} (left columns) and for Generalized RatioDCA (right columns).
Outliers were removed from the boxplots. The green lines are the values of modularity obtained using the linear method,
whereas the yellow dots are the values of modularity obtained using the linear eigenvector as starting point.}\label{fig:boxplot_local}
\end{figure}

\begin{figure}[t]
\centering
\subfloat[Linear]
{\includegraphics[scale=0.56, trim = 4.5cm 2cm 3cm 1.5cm, clip]{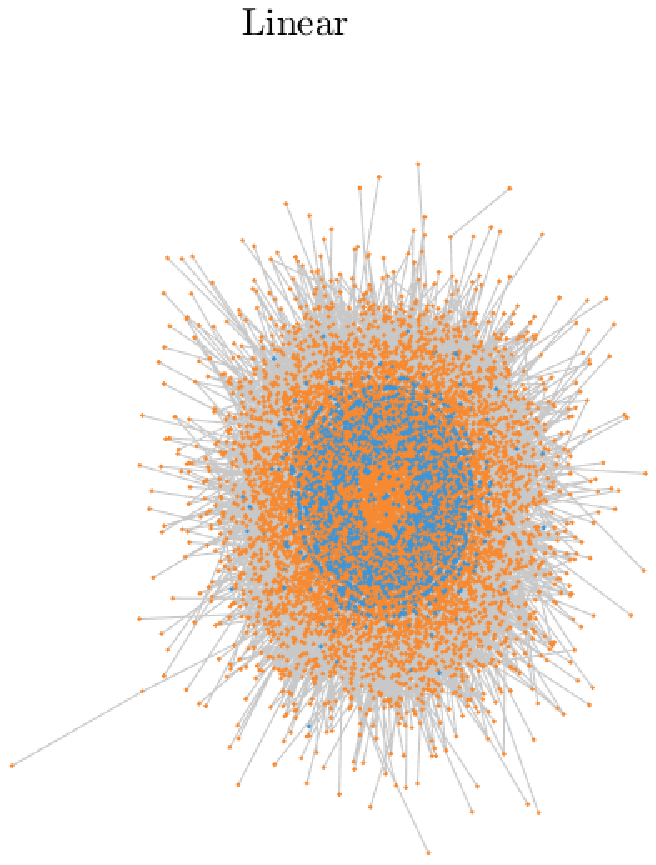}} \;
\subfloat[Generalized RatioDCA]
{\includegraphics[scale=0.56, trim = 4.5cm 2cm 3cm 1.5cm, clip]{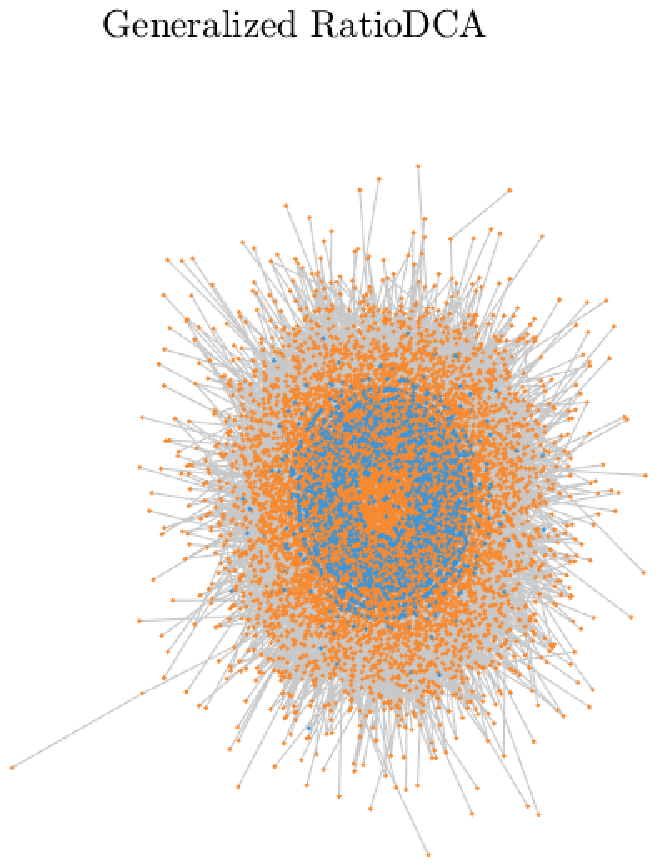}} \;
\subfloat[\FASTATVO]
{\includegraphics[scale=0.56, trim = 4.5cm 2cm 3cm 1.5cm, clip]{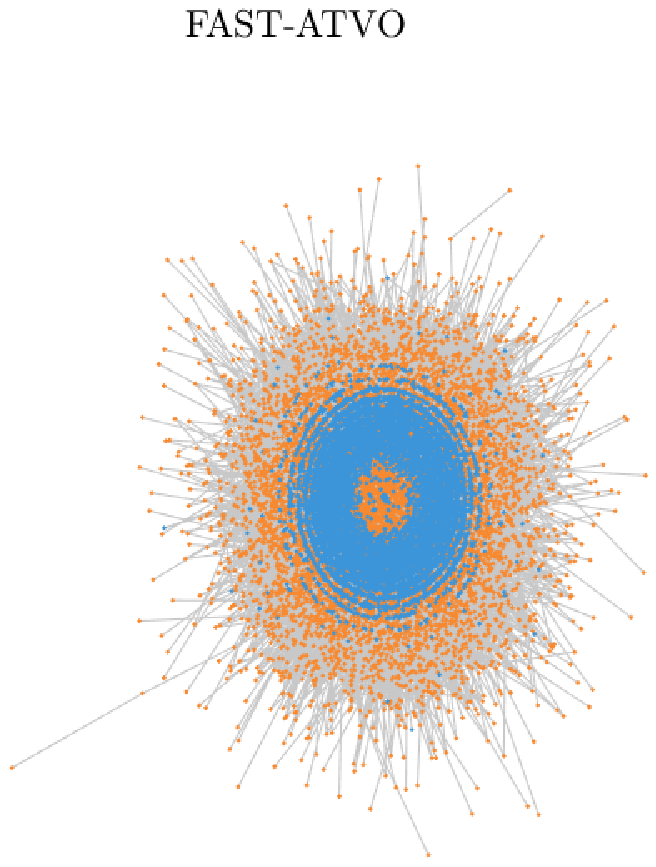}}
\caption{Graph drawing of the leading module obtained on the dataset number $7$ for the linear method (left),  Generalized RatioDCA (center) and the proposed \FASTATVO{} (right). In this case, Generalized RatioDCA returns a solution very close to the linear eigenvector assignment.
This might indicate that Generalized RatioDCA gets trapped in the basin of attraction of a local minimum.
\FASTATVO{}, instead, is able to escape that point and to find a much better leading module assignment.}\label{fig:community_example}
\end{figure}

Table \ref{tab:results_local_2} shows mean and standard deviations of the modularity values obtained by
\FASTATVO{} and  Generalized RatioDCA over 10 runs, each with a different randomly chosen starting point. This table confirms the behavior observed in Table \ref{tab:results_local_1}: the average modularity value obtained with \FASTATVO{} highly outperforms both the linear method and  Generalized RatioDCA. Moreover, the modularity standard deviations shown in the table demonstrate that \FASTATVO{} is generally more robust than Generalized RatioDCA.

\begin{table}[t!]
\centering
\caption{Performance comparison using random starting points. The values shown here correspond to the mean and standard deviation of $Q(S^*)$, where  $S^*$ is the community identified by optimal thresholding the output of the three methods: linear (\textbf{Q$_\text{linear}$}), Generalized RatioDCA ($\textbf{Q}_\textbf{R}$) and \FASTATVO{} ($\textbf{Q}_\textbf{F}$).}
{\scriptsize
{\begin{tabular}{| c | c c c c c | c c |}
\hline
\textbf{Dataset} \bigstrut[t]
& \multirow{2}*{\textbf{Q$_\text{linear}$}}
& \textbf{avg} & \textbf{std}
& \textbf{avg} & \textbf{std}
& \multirow{2}*{$\dfrac{\textbf{avg Q}_\textbf{F}}{\textbf Q_\textbf{linear}}$}
& \multirow{2}*{$\dfrac{\textbf{avg Q}_\textbf{F}}{\textbf{avg Q}_\textbf{R}}$} \bigstrut[b] \\
\textbf{ID} & & \textbf{Q$_\text{R}$} & \textbf{Q$_\text{R}$}
& \textbf{Q$_\text{F}$} & \textbf{Q$_\text{F}$} & & \bigstrut[b] \\
\hline
  1 &  0.28 &  0.11 &  0.06 &  0.38 &  0.01 &  1.39 &  3.41 \bigstrut[t] \\
  2 &  0.20 &  0.27 &  0.01 &  0.31 &  0.03 &  1.56 &  1.16 \\
  3 &  0.17 &  0.36 &  0.01 &  0.37 &  0.01 &  2.20 &  1.03 \\
  4 &  0.33 &  0.12 &  0.10 &  0.45 &  0.04 &  1.35 &  3.84 \\
  5 &  0.28 &  0.23 &  0.03 &  0.32 &  0.03 &  1.13 &  1.39 \\
  6 &  0.35 &  0.12 &  0.08 &  0.39 &  0.02 &  1.10 &  3.31 \\
  7 &  0.21 &  0.15 &  0.12 &  0.40 &  0.07 &  1.87 &  2.61 \\
  8 &  0.25 &  0.33 &  0.01 &  0.32 &  0.01 &  1.31 &  0.97 \\
  9 &  0.23 &  0.34 &  0.01 &  0.35 &  0.02 &  1.48 &  1.02 \\
 10 &  0.32 &  0.44 &  0.01 &  0.42 &  0.02 &  1.32 &  0.95 \\
 11 &  0.23 &  0.24 &  0.05 &  0.38 &  0.01 &  1.69 &  1.60 \\
 12 &  0.26 &  0.07 &  0.07 &  0.38 &  0.01 &  1.49 &  5.89 \\
 13 &  0.32 &  0.42 &  0.03 &  0.38 &  0.02 &  1.18 &  0.91 \\
 14 &  0.27 &  0.24 &  0.13 &  0.38 &  0.05 &  1.42 &  1.57 \\
 15 &  0.20 &  0.13 &  0.07 &  0.36 &  0.01 &  1.77 &  2.69 \\
 16 &  0.31 &  0.44 &  0.01 &  0.42 &  0.02 &  1.36 &  0.96 \bigstrut[b] \\
\hline
\end{tabular}}
}
\label{tab:results_local_2}
\end{table}

The results of Tables  \ref{tab:results_local_1} and \ref{tab:results_local_2} are summarized and shown together in the boxplots of Figure \ref{fig:boxplot_local}
(removing outliers), where we also compare the execution times of the methods. The upper panels in Subfigures~\ref{subfig:boxplot_local_1} and~\ref{subfig:boxplot_local_2} show medians and quartiles of the modularity values obtained with \FASTATVO{} (left columns) and Generalized RatioDCA (right columns) over 10 runs with random starting points, together with the value obtained using the linear eigenvector as a starting point (yellow dot). A straight green line shows, instead, the modularity value obtained with the linear spectral method. The lower panels in the subfigures show  median and quartiles of execution times of \FASTATVO{} and Generalized RatioDCA. As for the modularity values, we can see that \FASTATVO{} is generally more efficient (as it requires a smaller median execution time) and more robust (as the variance of the CPU time is in general remarkably smaller).

Finally, we report in Figure \ref{fig:plot_global} the performance comparison among \texttt{PS} framework,
\FASTATVO{} and Generalized RatioDCA.
For the three methods we use the linear eigenvector as starting point.
As we can easily see, the use of a global optimization strategy improves modularity values with respect to \FASTATVO{},
while guaranteeing good performances in terms of CPU time.
In particular, we notice that \texttt{PS} gives a modularity value higher than Generalized RatioDCA even for dataset \verb!Words28!,
which was the only one where \FASTATVO\ was outperformed by Generalized RatioDCA.

\begin{figure}[t]
\centering
\includegraphics[width=.7\textwidth, trim = 0cm 3.5cm 0cm 0cm, clip]{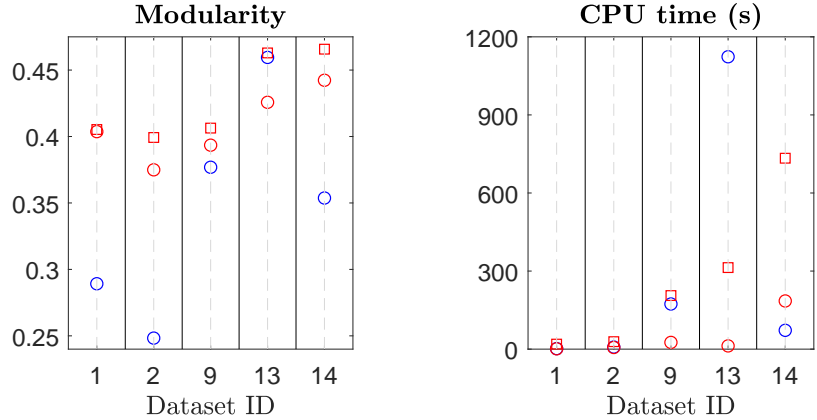}
\caption{Modularity values (left panel) and execution times (right panel) for \texttt{PS} (red squares)
\FASTATVO{} (red circles) and Generalized RatioDCA (blue circles). Linear eigenvector is used as starting point for all methods.}
\label{fig:plot_global}
\end{figure}

\section{Conclusions}\label{conclusions}
In this paper, we described a new modularity total variation approach for the leading community problem. We first proved that
the original combinatorial problem is equivalent to the minimization of a suitably chosen continuous nonsmooth function (the modularity total variation) over a box.
We then considered a smooth approximation of the continuous problem  and developed an algorithmic framework to efficiently tackle it.
The reported results show that the proposed method outperforms both Generalized RatioDCA and the linear method in terms of modularity value. Furthermore,
the CPU time required to find those good solutions is usually significantly smaller than the CPU time needed to run Generalized RatioDCA.
Hence, even if the modularity landscape is usually fraught with local optima (see, e.g., \cite{good2010performance}), the proposed
approach quickly finds solutions with good modularity values in the end. As far as we can see, the use of both the non-monotone line search and the randomized working set
selection seems to be of help in avoiding/escaping ``bad'' local optima.

Future work might focus on  adapting our strategy to efficiently handle applications where dense weighted graphs (such as similarity or correlation graphs) or multilayer graphs are required.
Another interesting research direction might be analyzing how the described approach transfers and applies to other community detection techniques like, e.g., Belief Propagation \cite{zhang2014scalable}, Stochastic Block Models \cite{abbe2017community} and Infomap~\cite{rosvall2009map}.

\appendix
\section{Detailed scheme and convergence analysis of \FASTATVO}\label{appendix}

In this appendix, we report the detailed scheme of \FASTATVO\ (a short scheme was given in Algorithm~\ref{alg:FAST_ATVO_SHORT}),
together with its convergence analysis. The detailed scheme is reported below in Algorithm~\ref{alg:FAST_ATVO}.

\begin{algorithm}[h!]
\caption{\FASTATVO($x^0$) -- detailed scheme}
\label{alg:FAST_ATVO}
\begin{algorithmic}\setlength{\itemsep}{2pt}
\scriptsize
\vspace{5pt}
\item[]$\,\,\,0$ Given a feasible point $x^0$, fix $Z\ge 1$, $M\ge 0$, $\Delta_0\ge0$, $\beta\in (0,1)$, $\delta\in (0,1)$,$\gamma\in (0,1)$, \\
            $\,\,\,\,\,\,$ $0 < \mu_\text{min} \le \mu_\text{max} < \infty$ and set $k = 0$,  $j = 0$, $l^0 = 0$, $f_R^0 = f^0 = f(x^0)$, $\Delta = \Delta_0$
\vspace{5pt}
\item[]$\,\,\,1$ While $x^k$ is a non-stationary point for problem~\eqref{prob}
\vspace{5pt}
\item[]$\,\,\,\,\,\,$\hspace{15pt} \textbf{Active and non-active set estimate}
\item[]$\,\,\,2$\hspace*{0.6truecm} Compute $A_l^k=A_l(x^k)$, $A_u^k=A_u(x^k)$ and $N^k=N(x^k)$
\vspace{5pt}
\item[]$\,\,\,\,\,\,$\hspace{15pt} \textbf{Function control every Z iterations}
\item[]$\,\,\,3$\hspace{15pt} If $k = l^j + Z$, then compute $f(x^k)$
\item[]$\,\,\,4$\hspace{30pt} If $f(x^k) \ge f^j_R$
\item[]$\,\,\,5$\hspace{45pt} Backtrack to $x^{l^j}$, set $d^k = d^{l^j}$, $k = l^j$ and go to step~22
\item[]$\,\,\,6$\hspace{30pt} Else
\item[]$\,\,\,7$\hspace{45pt} Set $j = j + 1$, $l^j = k$, $f^j = f(x^k)$ and $\displaystyle{f_R^j = \max_{0 \le i \le \min\{j,M\}} f^{j-i}}$
\item[]$\,\,\,8$\hspace{30pt} End if
\item[]$\,\,\,9$\hspace{15pt} End if
\vspace{5pt}
\item[]$\,\,\,\,\,\,$\hspace{15pt} \textbf{Computation of the search direction}
\item[]$10$\hspace{15pt} Find $\hat \imath^k \in \argmax_{i \in N^k} \abs{x^k_i - [x^k-\nabla f(x^k)]^{\sharp}_i}$
\item[]$11$\hspace{15pt} Choose a non-empty working set $W^k \subseteq N^k$ such that $\hat \imath^k \in W^k$
\item[]$12$\hspace{15pt} Set $d^k_{A_l^k} = 0$, $d^k_{A_u^k} = 0$, $d^k_{N^k \setminus W^k} = 0$ and set $d^k_{W^k} = -\frac 1{\mu^k} \nabla_{W^k} f(x^k)$,
                         with $\mu^k \in [\mu_\text{min}, \mu_\text{max}]$
\vspace{5pt}
\item[]$\,\,\,\,\,\,$\hspace{15pt} \textbf{Test for accepting the unit stepsize}
\item[]$13$\hspace{15pt} If $\norm{[x^k + d^k]^{\sharp}} \le \Delta$
\item[]$14$\hspace{30pt} Set $x^{k+1}=[x^k + d^k]^{\sharp}$, $\Delta = \beta \Delta$, $k = k + 1$ and go to step~1
\item[]$15$\hspace{15pt} Else if $k \ne l^j + Z$, then compute $f(x^k)$
\item[]$16$\hspace{30pt} If $f(x^k) \ge f^j_R$
\item[]$17$\hspace{45pt} Backtrack to $x^{l^j}$, set $d^k = d^{l^j}$, $k = l^j$ and go to step~22
\item[]$18$\hspace{30pt} Else
\item[]$19$\hspace{45pt} Set $j = j + 1$, $l^j = k$, $f^j = f(x^k)$ and $\displaystyle{f_R^j = \max_{0 \le i \le \min\{j,M\}} f^{j-i}}$
\item[]$20$\hspace{30pt} End if
\item[]$21$\hspace{15pt} End if
\vspace{5pt}
\item[]$\,\,\,\,\,\,$\hspace{15pt} \textbf{Non-monotone Armijo line search}
\item[]$22$\hspace{15pt} Set $\alpha^k = (\delta)^{\nu}$, where $\nu$ is the smallest non-negative integer such that
                         \begin{equation*}
                         f([x^k+\delta^{\nu} d^k]^{\sharp}) \le  f^j_R + \gamma (\delta)^{\nu} \nabla f(x^k)^\top d^k
                         \end{equation*}
\item[]$23$\hspace{15pt} Set $x^{k+1} = [x^k+\alpha^k d^k]^{\sharp}$ and $k = k + 1$
\vspace{5pt}
\item[]$24$ End while
\vspace{5pt}
\end{algorithmic}
\end{algorithm}

We see that every iteration $k$ starts with a non-stationary point $x^k$.
We first compute the active and non-active sets estimates as in step~2 and, if necessary, we evaluate the objective function (steps 3--9).
More in detail, thanks to a non-monotone stabilization strategy, inspired from that used in~\cite{grippo1986nonmonotone},
we can compute $f(x^k)$ only once every $Z$ iterations (if some tests described below are satisfied),
rather than at each iteration, allowing us to save computational time.
When we perform this function control, we compare $f(x^k)$ with a reference value $f^j_R$, which is the maximum among the last $M$ function evaluations:
if $f(x^k) \ge f^j_R$ we backtrack to the best point computed so far (i.e., $x^{l^j}$) and start a line search, otherwise we update $f^j_R$ and go on.

At steps~10--12, we choose a non-empty working set $W^k \subseteq N^k$ by a Gauss-Southwell-type (or greedy) rule: $W^k$ must contain
the index $\hat \imath^k$ of the variable that most violates stationarity. We then compute a search direction $d^k$ such that
\begin{equation}\label{d_W}
d^k_{W^k} = -\frac 1{\mu^k} \nabla_{W^k} f(x^k),
\end{equation}
with $\mu^k > 0$, and all the other components of $d^k$ are equal to zero.
We see that, at each iteration,  only a subset of variables (i.e., those in $W^k$) can be moved and
the search direction is computed in the variable subspace defined by $W^k$.
The computation of the coefficient $\mu^k$ will be described in details later on.

Afterwards, in steps~13--21, if $\norm{[x^k + d^k]^{\sharp}}$ if sufficiently small we set $x^{k+1}=[x^k + d^k]^{\sharp}$ and terminate the iteration.
It means that, if that test is satisfied, we accept the unit stepsize without computing the objective function.
Otherwise, we have two possibilities: if $f(x^k) \ge f^j_R$ we backtrack to $x^{l^j}$, else we update $f^j_R$ and go on.
Both these cases are followed by a non-monotone line search.

In the last steps~22--23 of the algorithm, we compute the stepsize $\alpha^k$ in order to set $x^{k+1} = [x^k+\alpha^k d^k]^{\sharp}$.
We use a non-monotone Armijo line search with reference value equal to $f^j_R$ (which was first proposed in~\cite{grippo1986nonmonotone}).

Now, let us describe how we compute the search direction $d^k$ (step~12).
As mentioned before, we choose a coefficient $\mu^k$ (in a finite positive interval) and compute $d^k_{W^k}$ as in~\eqref{d_W},
while all the other components of $d^k$ are set to zero. In this way, at each iteration $k$ we move only the variables in $W^k$.
In our experiments, we computed $d^k_{W^k}$ as a spectral (or Barzilai-Borwein) gradient direction
(see, e.g.,~\cite{fletcher2005barzilai} and the references therein).

More precisely, based on the strategy proposed in~\cite{andreani2010second,birgin2002large}, for $k < 2$ we set
$$\displaystyle{\mu^k = \max\Biggl\{\mu_\text{min}, \, \min\biggl\{1, \, \dfrac{\norm{x^k_{W^k}}}{\norm{\nabla_{W^k} f(x^k)}}\biggr\}\Biggr\}}$$
and, for $k \ge 2$,
\begin{equation*}
\mu^k =
\begin{cases}
\max\{\mu_\text{min}, \, \mu^k_a\}, \quad & \text{if } 0 < \mu^k_a < \mu_\text{max}, \\[1.1ex]
\max\bigl\{\mu_\text{min}, \, \min\{\mu_\text{max}, \, \mu^k_b\}\bigr\}, \quad & \text{if } \mu^k_a \ge \mu_\text{max}, \\[1.1ex]
\max\Biggl\{\mu_\text{min}, \, \min\biggl\{1, \, \dfrac{\norm{x^k_{W^k}}}{\norm{\nabla_{W^k} f(x^k)}}\biggr\}\Biggr\}, \quad & \text{if } \mu^k_a \le 0,
\end{cases}
\end{equation*}
with $0 < \mu_\text{min} \le \mu_\text{max} < \infty$,
$\mu^k_a = \dfrac{(s^{k-1})^\top y^{k-1}}{\norm{s^{k-1}}^2}$, $\mu^k_b = \dfrac{\norm{y^{k-1}}^2}{(s^{k-1})^\top y^{k-1}}$,
$s^{k-1} = x^k_{W^k}-x^{k-1}_{W^k}$ and $y^{k-1} = \nabla_{W^k} f(x^k) - \nabla_{W^k} f(x^{k-1})$.
\subsection{Global convergence proof}
In order to prove the global convergence of \FASTATVO{} we need some preliminary results that we derive below. Throughout the whole section, we use the notation of Algorithm \ref{alg:FAST_ATVO}.

First, we report a known result of non-monotone methods, whose proof can be easily adapted from the proof of Lemma~3 of~\cite{grippo1991class}
and is omitted here for the sake of brevity. 
\begin{lemma}\label{lemma:nonmonotone}
Assume that $\{x^k\}$ is an infinite sequence of points produced by \FASTATVO. Then,
\begin{gather}
\label{conv_f_R} \lim_{k \to \infty} f(x^{k}) = \lim_{j \to \infty} f^j_R = \bar f_R \in \R, \\
\label{conv_x} \lim_{k \to \infty} \norm{x^{k+1} - x^k} = 0.
\end{gather}
\end{lemma}
Further, we state the following technical lemma
\begin{lemma}\label{lem:subsequences}
There exist subsequences $\{\alpha^k\}_K$, $\{x^k\}_K$, $\{d^k\}_K$ and sets $\bar A_l$, $\bar A_u$, $\bar N$, $\bar W$ such that
\begin{gather*}
\lim_{k \to\infty, \, k \in K} \alpha^k = \bar \alpha \in \R, \\
\lim_{k \to\infty, \, k \in K} x^k = \bar x \in \R^n, \\
\lim_{k \to\infty, \, k \in K} d^k = \bar d \in \R^n, \\
A_l^k = \bar A_l, \quad A_u^k = \bar A_u, \quad N^k = \bar N, \quad W^k = \bar W, \quad \forall k \in K,
\end{gather*}
with $\bar \alpha >0$ and $\norm{\bar d} > 0$.
\end{lemma}
\begin{proof}
The thesis is an immediate consequence of the fact that $A^k_l$, $A^k_u$, $N^k$, $W^k$ are subsets of a finite set of indices and the sequences $\{\alpha^k\}$, $\{x^k\}$, $\{d^k\}$ are bounded. In particular, the latter property itself follows from the fact that $0 \le \alpha^k \le 1$,
from the compactness of the feasible set and from the compactness of the feasible set combined with the definition of $d^k$ and the continuity of $\nabla f$, respectively.
\end{proof}
Using the previous lemmas, we derive the following result.
\begin{lemma}\label{lemma:alphad_to_zero}
Assume that $\{x^k\}$ is an infinite sequence of points produced by \FASTATVO. Then,
\begin{equation*}
\lim_{k \to \infty} \alpha^k \norm{d^k} = 0.
\end{equation*}
\end{lemma}
\begin{proof}
We proceed by contradiction and we assume that the result is not true.
Let  $\{\alpha^k\}_K$, $\{x^k\}_K$, $\{d^k\}_K$ and $\bar A_l$, $\bar A_u$, $\bar N$, $\bar W$ be as in Lemma \ref{lem:subsequences}.
Using~\eqref{conv_x} and the continuity of the projection operator, we can thus write
\begin{equation}\label{lim_x}
\lim_{k \to\infty, \, k \in K} (x^{k+1}_i - x^k_i) = \lim_{k \to\infty, \, k \in K} ([\bar x + \bar \alpha \bar d]^{\sharp}_i - \bar x_i) = 0, \quad i = 1,\ldots,n.
\end{equation}
In the sequel, we show that the above limit leads to $\norm{\bar d} = 0$, getting a contradiction.
To this extent, let us first observe that, from the definition of $d^k$ and the continuity of $\nabla f$, we have
\begin{equation}\label{d_W_proof}
\bar d_{\bar W} = - \frac 1 {\bar \mu} \nabla_{\bar W} f(\bar x),
\end{equation}
for some $\bar \mu > 0$.
Moreover, from our estimates~\eqref{active_set_estimates} and the continuity of $\nabla f$, there exists an iteration $\hat k$ such that
\begin{subequations}
\begin{align}
\bar x_i = l_i, \, \nabla_i f(\bar x) > 0 \quad \Rightarrow \quad i \notin \bar N, \; \forall k \ge \hat k, \, k \in K, \label{i_in_A_l} \\
\bar x_i = u_i, \, \nabla_i f(\bar x) < 0 \quad \Rightarrow \quad i \notin \bar N, \; \forall k \ge \hat k, \, k \in K. \label{i_in_A_u}
\end{align}
\end{subequations}
Now, let us consider any index $i \in \{1,\ldots,n\}$. We can distinguish four possible cases.
\begin{enumerate}[label=(\roman*), leftmargin=*]
\item $i \notin \bar W$. From the definition of the search direction, it follows that $\bar d_i = 0$.
\item $i \in \bar W$ such that $l_i < \bar x_i < u_i$. From~\eqref{lim_x} and the fact that $\bar \alpha > 0$, we obtain $\bar d_i = 0$.
\item $i \in \bar W$ such that $\bar x_i = l_i$. Recalling that $\bar W \subseteq \bar N$, from~\eqref{i_in_A_l} and the definition of $d^k$ we have $\nabla_i f(\bar x) \le 0$.
    Using~\eqref{d_W_proof} we obtain $\bar d_i \ge 0$, which, combined with~\eqref{lim_x} and the fact that $\bar \alpha > 0$, implies that $\bar d_i = 0$.
\item $i \in \bar W$ such that $\bar x_i = u_i$. Reasoning as above, we get $\bar d_i = 0$.
\end{enumerate}
We thus obtain $\norm{\bar d} = 0$, leading to a contradiction.
\end{proof}

We can now show that the sequence of directional derivatives $\{\nabla f(x^k)^\top d^k\}$ converges to zero,
which will be crucial to prove the global convergence of the algorithm.
\begin{proposition}\label{prop:lim_dir_der}
Assume that $\{x^k\}$ is an infinite sequence of points produced by \FASTATVO. Then,
\begin{equation}\label{lim_gd}
\lim_{k \to \infty} \nabla f(x^k)^\top d^k = 0.
\end{equation}
\end{proposition}

\begin{proof}
By contradiction, we assume that~\eqref{lim_gd} does not hold.
There must exist $\{\alpha^k\}_K$, $\{x^k\}_K$, $\{d^k\}_K$ and $\bar A_l$, $\bar A_u$, $\bar N$, $\bar W$ defined as in Lemma \ref{lem:subsequences},
which also satisfy
\begin{equation}\label{lim_gd_contr}
\lim_{k \to\infty, \, k \in K} \nabla f(\bar x)^\top \bar d = -\eta < 0.
\end{equation}
Combining~\eqref{lim_gd_contr} with Lemma~\ref{lemma:alphad_to_zero}, we obtain
\begin{equation}\label{alphatozero}
\lim_{k\to\infty,\, k\in K} \alpha^k = 0.
\end{equation}
Therefore, there exist two further infinite subsequences, that, with a slight abuse of notation, we still denote by $\{x^k\}_K$ and $\{d^k\}_K$,
such that $\alpha^k < 1$ for all $k \in K$. From Algorithm \ref{alg:FAST_ATVO}, $\alpha^k$ can be less than $1$ if and only if
it is computed by the non-monotone Armijo line search using $\nu \ge 1$ (see steps~22--23). It follows that
\begin{equation}\label{arm_not_satisf1}
\begin{split}
f\bigl(\bigl[x^k + \frac{\alpha^k}{\delta}d^k\bigr]^{\sharp}\bigr) & > f^{q(k)}_R + \gamma \frac{\alpha^k}{\delta} \nabla f(x^k)^Td^k \\
& \ge f(x^k) + \gamma \frac{\alpha^k}{\delta} \nabla f(x^k)^Td^k, \quad \forall k \in K,
\end{split}
\end{equation}
where $q(k) = \max \{ j \colon l^j \le k \}$.
Let us write the point $[x^k + \frac{\alpha^k}{\delta}d^k]^{\sharp}$ as follows:
\begin{equation}\label{x_armijo_rej}
\bigl[x^k + \frac{\alpha^k}{\delta}d^k\bigl]^{\sharp} = x^k + \frac{\alpha^k}{\delta}d^k - y^k,
\end{equation}
where
\begin{equation}\label{y_k}
y^k_i =
\begin{cases}
x^k_i + \dfrac{\alpha^k}{\delta}d^k_i - l_i, \quad & \text{if } x^k_i + \dfrac{\alpha^k}{\delta}d^k_i < l_i \\
x^k_i + \dfrac{\alpha^k}{\delta}d^k_i - u_i, \quad & \text{if } x^k_i + \dfrac{\alpha^k}{\delta}d^k_i > u_i \\
0, \quad                                           & \text{otherwise},
\end{cases}
\end{equation}
or equivalently,
\begin{equation*}
y^k_i = \max\bigl\{0, \bigl(x^k + \frac{\alpha^k}{\delta}d^k\bigr)_i - u_i\bigr\} - \max\bigl\{0,l_i-\bigl(x^k + \frac{\alpha^k}{\delta}d^k\bigr)_i\bigr\}, \quad i=1,\dots,n.
\end{equation*}
Using~\eqref{alphatozero}, the feasibility of points $x^k$ and the fact that $\{d^k\}$ is bounded, we have
\begin{equation}\label{lim_y_k}
\lim_{k\to\infty, \, k\in K} y^k = 0.
\end{equation}
Now, from~\eqref{arm_not_satisf1} and~\eqref{x_armijo_rej} we can write
\begin{equation}\label{arm_not_satisf2}
f\bigl(x^k + \frac{\alpha^k}{\delta}d^k - y^k\bigr) - f(x^k) > \gamma \frac{\alpha^k}{\delta} \nabla f(x^k)^\top d^k, \quad \forall k \in K.
\end{equation}
By the mean value theorem, we also have
\begin{equation}\label{mean_th_point}
f\bigl(x^k + \frac{\alpha^k}{\delta}d^k - y^k\bigr) = f(x^k) + \frac{\alpha^k}{\delta}\nabla f(z^k)^\top d^k - \nabla f(z^k)^\top y^k,
\end{equation}
where $z^k = x^k + \theta^k\bigl(\alpha^k d^k/\delta - y^k\bigr)$ and $\theta^k \in (0,1)$.
Using~\eqref{alphatozero}, \eqref{lim_y_k} and the fact that $\{d^k\}$ is bounded, we obtain
\begin{equation}\label{lim_z}
\lim_{k\to\infty, \, k\in K} z^k = \bar x.
\end{equation}
Moreover, from~\eqref{arm_not_satisf2} and~\eqref{mean_th_point}, we have
\begin{equation}\label{arm_not_satisf3}
\nabla f(z^k)^\top d^k - \frac{\delta}{\alpha^k}\nabla f(z^k)^\top y^k > \gamma \nabla f(x^k)^\top d^k, \quad \forall k \in K.
\end{equation}

In the following, we show that
\begin{equation}\label{liminf_major}
\liminf_{k\to\infty, \, k\in K} \frac{\delta}{\alpha^k}\nabla f(z^k)^\top y^k \ge 0.
\end{equation}
From~\eqref{y_k} and the fact that each $x^k$ is feasible, we first observe that
\begin{equation}\label{y_k_sign}
y^k_i
\begin{cases}
\in [\alpha^k d^k_i/\delta, 0], \quad & \text{if } d^k_i < 0, \\
\in [0, \alpha^k d^k_i/\delta], \quad  & \text{if } d^k_i > 0, \\
= 0, \quad                             & \text{if } d^k_i = 0.
\end{cases}
\end{equation}
Now, we consider any index $i \in \{1,\ldots,n\}$ and we analyze four possible cases, in order to show that~\eqref{liminf_major} holds.
\begin{enumerate}[label=(\roman*), leftmargin=*]
\item $i \in \{1,\ldots,n\} \setminus \bar W$. From the rule used to compute $d^k$, we have $d^k_i = 0$.
    Using~\eqref{y_k_sign} it follows that $y^k_i = 0$ for all $k \in K$, implying that
    \begin{equation}\label{y_zero_0}
    \lim_{k\to\infty, \, k\in K} \frac{\delta}{\alpha^k}\nabla_i f(z^k) y^k_i = 0.
    \end{equation}
\item $i \in \bar W$ such that $l_i < \bar x_i < u_i$. Since $\{x^k\}_K \to \bar x$, for all sufficiently large $k \in K$ we have
    $l_i + \tau \le x^k \le u_i - \tau$ for some $\tau > 0$.
    Using~\eqref{alphatozero} and the fact that $\{d^k\}$ is bounded, for all sufficiently large $k \in K$ we have
    \begin{equation*}
    l_i < x^k_i + \frac{\alpha^k}{\delta}d^k_i < u_i,
    \end{equation*}
    which, combined with~\eqref{y_k}, implies that $y^k_i = 0$ for all sufficiently large $k \in K$.
    Then,
    \begin{equation}\label{y_zero_1}
    \lim_{k\to\infty, \, k\in K} \frac{\delta}{\alpha^k}\nabla_i f(z^k) y^k_i = 0.
    \end{equation}
\item $i \in \bar W$ such that $\bar x_i = l_i$. Since $\{x^k\}_K \to \bar x$ and $\nabla f$ is continuous,
    from our estimates~\eqref{active_set_estimates} (and recalling that $\bar W \subseteq \bar N$) we have
    \begin{equation}\label{g_i_x}
    \nabla_i f(\bar x) \le 0.
    \end{equation}
    Now, we also show that there exists $\hat k \in K$ such that
    \begin{equation}\label{y_nonpos}
    y^k_i \le 0, \quad \forall k \ge \hat k, \, k \in K.
    \end{equation}
    Indeed, since $\{x^k\}_K \to \bar x$, for all sufficiently large $k \in K$ we have
    $x^k \le u_i - \tau$ for some $\tau > 0$.
    Using~\eqref{alphatozero} and the fact that $\{d^k\}$ is bounded, for all sufficiently large $k \in K$ we have
    \begin{equation*}
    x^k_i + \frac{\alpha^k}{\delta}d^k_i < u_i,
    \end{equation*}
    which, combined with~\eqref{y_k}, implies~\eqref{y_nonpos}.
    Let us partition $K$ into $K_1$ and $K_2$, such that every $k \in K$ belongs to $K_1$ if and only if $d^k_i \ge 0$
    (and then, every $k \in K$ belongs to $K_2$ if and only if $d^k_i < 0$).
    Assuming without loss of generality that both $K_1$ and $K_2$ are infinite, we now analyze the corresponding subsequences.
    \begin{itemize}
    \item For what concerns $K_1$, using~\eqref{y_k_sign} we have that $y^k_i \ge 0$ for all $k \in K_1$.
        From~\eqref{y_nonpos} it follows that $y^k_i = 0$ for all sufficiently large $k \in K_1$. Then,
        \begin{equation}\label{y_zero_2}
        \lim_{k\to\infty, \, k \in K_1} \frac{\delta}{\alpha^k}\nabla_i f(z^k) y^k_i = 0.
        \end{equation}
    \item For what concerns $K_2$, taking into account~\eqref{g_i_x} we distinguish two possible cases.
        \begin{enumerate}[label=(\alph*), leftmargin=*]
        \item $\nabla_i f(\bar x) < 0$.
            Since $\{z^k\}_K \to \bar x$, then $\nabla_i f(z^k) < 0$ for all sufficiently large $k \in K_2$.
            From~\eqref{y_nonpos} it follows that $\nabla f_i(z^k) y^k_i \ge 0$ for sufficiently large $k \in K_2$.
            Then,
            \begin{equation}\label{y_zero_3}
            \liminf_{k\to\infty, \, k \in K_2} \frac{\delta}{\alpha^k}\nabla_i f(z^k) y^k_i \ge 0.
            \end{equation}
        \item $\nabla_i f(\bar x) = 0$. From~\eqref{y_k_sign} we have $|y^k_i| \le \frac{\alpha^k}{\delta}|d^k_i|$.
            Therefore,
            \begin{equation*}
            0 < \frac{\delta}{\alpha^k} |\nabla_i f(z^k)y^k_i| \le \frac{\delta}{\alpha^k} |\nabla_i f_i(z^k)| |y^k_i| \le |\nabla_i f(z^k)| |d^k_i|.
            \end{equation*}
            Since $\{z^k\}_K \to \bar x$ and $\{d^k\}$ is bounded, from the continuity of $\nabla f$ we get
            \begin{equation}\label{y_zero_4}
            \lim_{k\to\infty, \, k \in K_2} \frac{\delta}{\alpha^k} \nabla_i f(z^k)d^k_i = 0.
            \end{equation}
        \end{enumerate}
        From~\eqref{y_zero_2},\eqref{y_zero_3} and~\eqref{y_zero_4} we obtain that, for all $i \in \bar W$ such that $\bar x_i = l_i$,
        \begin{equation}\label{y_zero_5}
        \liminf_{k\to\infty, \, k \in K} \frac{\delta}{\alpha^k}\nabla_i f(z^k) y^k_i \ge 0.
        \end{equation}
    \end{itemize}
\item $i \in \hat N$ such that $\bar x_i = u_i$. Reasoning as in the previous case, we obtain
    \begin{equation}\label{y_zero_6}
    \liminf_{k\to\infty, \, k \in K} \frac{\delta}{\alpha^k}\nabla_i f(z^k) y^k_i \ge 0.
    \end{equation}
\end{enumerate}
Therefore,~\eqref{liminf_major} follows from~\eqref{y_zero_0}, \eqref{y_zero_1}, \eqref{y_zero_5} and~\eqref{y_zero_6}.
Combining~\eqref{liminf_major} with~\eqref{arm_not_satisf3}, we can write
\begin{equation*}
\begin{split}
0 & \le \liminf_{k\to\infty, \, k \in K} \bigl(\nabla f(z^k)^\top d^k - \frac{\delta}{\alpha^k}\nabla f(z^k)^\top y^k - \gamma \nabla f(x^k)^\top d^k\bigr) \\
  & \le \liminf_{k\to\infty, \, k \in K} \bigl(\nabla f(z^k)^\top d^k - \gamma \nabla f(x^k)^\top d^k\bigr) \\
  & = \lim_{k\to\infty, \, k \in K} \bigl(\nabla f(z^k)^\top d^k - \gamma \nabla f(x^k)^\top d^k\bigr) = (1-\gamma) \nabla f(\bar x)^\top \bar d,
\end{split}
\end{equation*}
where the last two equalities follow from the continuity of $\nabla f$ and the fact that both $\{x^k\}_K$ and $\{z^k\}_K$ converge to $\bar x$.
Using~\eqref{lim_gd_contr}, we finally obtain $(\gamma - 1)\eta \ge 0$, with $\eta > 0$ and $\gamma \in (0,1)$, which leads to a contradiction.
\end{proof}

We are finally ready to prove the global convergence of \FASTATVO\ to stationary points.
\begin{proof}[Proof of Theorem~\ref{thm:conv}]
Assume that the sequence $\{x^k\}$ generated by \FASTATVO\ is infinite and let $x^*$ be a limit point of $\{x^k\}$.
Further, let  $\{\alpha^k\}_K$, $\{x^k\}_K$, $\{d^k\}_K$ and $\bar A_l$, $\bar A_u$, $\bar N$, $\bar W$ be as in Lemma \ref{lem:subsequences}.
Using the stationarity conditions~\eqref{stationarity}, we can measure the stationarity violation of any feasible point $x$ by the following functions:
\begin{equation*}
\phi(x_i) = \min \bigl\{ \max\{l_i-x_i, -\nabla_i f(x)\}^2, \max\{x_i-u_i, \nabla_i f(x)\}^2 \bigr\}, \quad i=1,\dots, n.
\end{equation*}
Namely, a feasible point $x$ is stationary if and only if $\phi(x_i) = 0$ for all $i = 1,\ldots,n$.
Now, arguing by contradiction, assume that $x^*$ is non-stationary.
Then, an index $i$ such that $\phi(x^*_i) > 0$ exists.
Note that it must hold that
\begin{equation}\label{violation0}
\nabla_i f(x^*) \ne 0.
\end{equation}
Moreover, from the continuity of $\phi$, there exist $\epsilon \in (0,1)$ and $\hat k \in K$ such that
\begin{equation}\label{violation}
\phi(x^k_i) \ge \epsilon, \quad \forall k \ge \hat k.
\end{equation}
Now, we consider four possible cases.
\begin{enumerate}[label=(\roman*)]
\item $i \in \bar A_l$. From~\eqref{Al} we have $x^k_i = l_i$ and $\nabla_i f(x^k) > 0$ for all $k \in K$.
    Since $\{x^k\}_K \to x^*$, from the continuity of $\nabla f$ it follows that, for all sufficiently large $k \in K$,
    \begin{equation*}
    \nabla_i f( x^k) \ge -\frac{\epsilon}2.
    \end{equation*}
    Then, we have $\phi(x^k_i) \le \epsilon^2/4 < \epsilon$ for all sufficiently large $k \in K$. This contradicts~\eqref{violation}.
\item $i \in \bar A_u$. From~\eqref{Au} we have $x^k_i = u_i$ and $\nabla_i f(x^k) < 0$ for all $k \in K$.
    Then, we obtain a contradiction by the same arguments used above.
\item $i \in \bar W$. Using Proposition~\eqref{prop:lim_dir_der}, we have $\nabla f(x^*)^\top \bar d = 0$.
    From the definition of $d^k$ and the continuity of $\nabla f$, we obtain
    \begin{equation*}
    0=\nabla f(x^*)^\top \bar d = \nabla_{\bar W} f(x^*)^\top \bar d_{\bar W} = -\frac{1}{\bar\mu}\norm{\nabla_{\bar W} f(x^*)}^2 ,
    \end{equation*}
    and then $\nabla_i f(x^*) = 0$, contradicting~\eqref{violation0}.
\item $i \in \bar N \setminus \bar W$. Without loss of generality, we can assume that the index $\hat \imath^k$ computed at step~10 is constant
    and equal to $\hat \imath $ for all $k \in K$ (passing into a subsequence if necessary).
    So, using the definition of $\hat \imath^k$, for all $k \in K$ we can write
    \begin{equation*}
    \abs{x^k_i - [x^k-\nabla f(x^k)]^{\sharp}_i} \le \abs{x^k_{\hat \imath^k} - [x^k-\nabla f(x^k)]^{\sharp}_{\hat \imath^k}}
                                                 = \abs{x^k_{\hat \imath} - [x^k-\nabla f(x^k)]^{\sharp}_{\hat \imath}}.
    \end{equation*}
    Since $x^*_{\hat \imath}$ does not violate stationarity (from the fact that $\hat \imath \in \bar W$, as analyzed above),
    the continuity of the projection operator and the continuity of $\nabla f$ imply that
    \begin{equation*}
    \begin{split}
   0\leq \abs{x^*_i - [x^*-\nabla f(x^*)]^{\sharp}_i} & = \lim_{k \to \infty, \, k \in K} \abs{x^k_i - [x^k-\nabla f(x^k)]^{\sharp}_i} \\
                                                 & \le \lim_{k \to \infty, \, k \in K} \abs{x^k_{\hat \imath} - [x^k-\nabla f(x^k)]^{\sharp}_{\hat \imath}} \\
                                                 & = \abs{x^*_{\hat \imath} - [x^*-\nabla f(x^*)]^{\sharp}_{\hat \imath}} = 0.
    \end{split}
    \end{equation*}
   We hence have $\phi(x^*_i) = 0$, which gives a contradiction.
\end{enumerate}
Therefore $x^*$ must be a stationary point.
\end{proof}

\bibliography{references}

\end{document}